\documentclass[article]{revtex4-1}
\usepackage{graphicx}
\usepackage{epsfig,color}
\usepackage{amsmath,amssymb,amsthm}
\usepackage[utf8]{inputenc}
\usepackage[colorlinks, urlcolor=blue, linkcolor=red]{hyperref}
\usepackage{mathtools}
\usepackage{graphicx}

\newtheorem{lemma}{Lemma}
\newtheorem{prop}{Proposition}

\usepackage{soul}

\newcommand{\beqa}{\begin{eqnarray}}
\newcommand{\eeqa}{\end{eqnarray}}

\newcommand{\braket}[1]{\langle #1 \rangle}

\renewcommand{\boxed}[2]{\textcolor{#1}{%
\tikz[baseline={([yshift=-1ex]current bounding box.center)}] \node [rectangle, minimum width=1ex,rounded corners,draw] {\normalcolor\m@th$\displaystyle#2$};}}

\newcounter{appsection}
\newcounter{appsubsection}[appsection]

\usepackage{bookmark}

\begin{document}

\title{Gaussian quantum estimation of the lossy parameter in a thermal environment}
\author{Robert Jonsson}
\email{robejons@chalmers.se}
\affiliation{Department of Microtechnology and Nanoscience, Chalmers University of Tech., Göteborg Sweden}
\affiliation{New Concepts and System Studies, Surveillance, Saab, Göteborg Sweden}
\author{Roberto Di Candia}
\email{rob.dicandia@gmail.com}
\affiliation{Department of Communications and Networking, Aalto University, Espoo, 02150 Finland}

\begin{abstract}
Lossy bosonic channels play an important role in a number of quantum information tasks, since they well approximate thermal dissipation in an experiment. Here, we characterize their metrological power in the idler-free and entanglement-assisted cases, using respectively single- and two-mode Gaussian states as probes. In the problem of estimating the lossy parameter, we study the energy-constrained quantum Fisher information (QFI) for generic temperature and lossy parameter regimes, showing qualitative behaviours of the optimal probes. We show semi-analytically that the two-mode squeezed-vacuum state optimizes the QFI for any value of the lossy parameter and temperature. We discuss the optimization of the {\it total} QFI, where the number of probes is allowed to vary by keeping the total energy-constrained. In this context, we elucidate the role of the ``shadow-effect'' for reaching a quantum advantage. We also consider a photon-number normalization for the environment, widely used in the analysis of quantum illumination and quantum reading protocols. With this normalization, we prove that the large bandwidth TMSV state is the optimal probe for any parameter value. Here, the quantum advantage is of at most a factor of $2$, and is reached in the bright environment case for {\it any} lossy parameter values. Finally, we discuss the implications of our results for quantum illumination and quantum reading applications.
\end{abstract}

\maketitle

\section{Introduction}

Lossy channels are important to describe realistic scenarios in all quantum information tasks. A key example is the dissipative bosonic channels~\cite{Petruccione:book}. Assume a bosonic mode interacting with a thermal bath at a certain temperature. How is the quantum state susceptible to the presence of the bath? In other words, how well can we estimate the amount of losses given a certain probe? This question, aside being interesting for calibrating a number of physical setups, is important for  many imaging~\cite{Tsang2016, Nair2016, Lupo2016, Gregory2020}, detection~\cite{Tan2008, Pirandola2011, Lu2018, Barzanjeh2015, LasHeras2017, Sanz2017}, and communication~\cite{Shapiro2009, Bash2015, Rosati2018, Noh2020, DiCandia2018,DiCandia2021} scenarios. 
Quantum information tools based on the quantum Fisher information (QFI) have been developed in a general quantum parameter estimation framework. Mostly, one aims to answer questions about optimality of the input and the measurement. This is indeed challenging when the dynamics are non-unitary, because  the procedure involves computing distances and/or fidelities between mixed quantum states. However, the single loss parameter case is ``simple'' enough, and various aspects have been studied in the literature. Furthermore, the problem can be further simplified if one restricts the analysis to Gaussian probes~\cite{Banchi2015, Safranek2017, Serafini:book}.

There are various contributions tackling different aspects of the lossy parameter estimation problem, see Ref.~\cite{Braun2018} for a review. A first result is given by Sarovar and Milburn, who developed a general theory for finding the optimal estimator given a probe, with an application for the damping channel for a Fock state as input~\cite{Sarovar2006}. Venzl and Freyberger first noticed that the quantum estimation of the loss parameter can be improved using entanglement~\cite{Venzl2007}, but they limit their theory to superposition of coherent states with an unoptimized measurement. Monras and Paris proposed the first complete study of the optimal QFI with a generic Gaussian state input~\cite{Monras2007}. Their study has been extended to  non-Gaussian probes by Adesso \textit{et al.}~\cite{Adesso2009}. All these contributions have been developed in the zero temperature case. An extension of these results to the finite temperature and the entanglement-assisted cases has been advanced in Refs.~\cite{Monras2010,Monras2011}. More recently, a general theory for the estimating {\it multiple} loss parameters in zero temperature bath considering generic non-Gaussian states was recently introduced by Nair~\cite{Nair2018}. Here, the author found that states diagonal in the Fock basis are optimal. The result directly implies that, when restricting to Gaussian probes, two-mode squeezed-vacuum (TMSV) states are optimal for the estimation of the single loss parameter. Finally, extensions to non Gaussian-preserving models have been considered lately by Rossi {\it et al.} in Ref.~\cite{Rossi2016}, where the authors showed that the presence of a Kerr non-linearity can improve the estimation performance, especially at short-interaction times. Despite the numerous literature in the topic, a complete characterization of the optimal states when restricting to the single- and two-mode cases is still missing. 

In this article, we study the QFI for the estimation of the single loss parameter in the case of {\it thermal} channel of generic temperature. We provide analytical results about the optimal probe for {\it any} parameter regime. Indeed, we provide a rigorous analysis of the behaviour of the optimal probe in various energy regimes, for both the idler-free (\textit{i.e.}, single-mode probe) and the entanglement-assisted (or ancilla-assisted) cases. We complement our analytical results with exact numerical calculations. Our results departs from previous analysis, especially from Refs.~\cite{Monras2007, Monras2011,Nair2018}, in the following: (i) In the zero bath-temperature case, we provide analytical results for the behaviour of the optimal single-mode state. In particular, we characterize the requirements for the squeezed-vacuum and coherent states to be optimal, complementing the analysis in Ref.~\cite{Monras2007}. (ii) In the finite bath-temperature case, we show the presence of an abrupt transition  of the optimal probe between squeezed-vacuum and coherent states, at the low-energy regime. This transition disappears when the energy gets higher, and was not shown in Ref.~\cite{Monras2011}. (iii) We provide an analysis of the {\it total} QFI. In the zero-temperature case, we show that squeezed-vacuum states are optimal over a larger value-set of parameters when allowing the number of probes (or the bandwidth) to vary, while keeping the total energy-constrained. We also provide a first proof that the optimal setup consists in distributing the energy either on one probe or on an infinite number of probes, depending on the probe energy. We extend the total QFI analysis to the finite bath-temperature case, by introducing a normalization of the environmental photon-number widely used in quantum illumination and quantum reading protocols. (iv) We show semi-analytically that the TMSV state is optimal for any bath-temperature. This complements the optimality result in Ref.~\cite{Nair2018} for the zero temperature case. We extend the optimality proof for the normalized model given in Ref.~\cite{Nair2020}, showing that the infinite bandwidth TMSV state is an optimal probe for arbitrary values of the lossy parameter. 
Finally, we show the relation to the task of discriminating between two values of the lossy parameter. We discuss the implications of our findings for the performance of two important protocols: quantum illumination~\cite{Tan2008} and quantum reading~\cite{Pirandola2011}. In particular, we discuss the qualitative difference between the normalized and unnormalized models, showing a discrepancy both in the QFI behaviour and the optimal receivers in relevant regimes of the input power and lossy parameters.

The paper is structured in the following way. We first introduce the notations via a Setup and Methods section (Section~\hyperref[sec:II]{II}), where we describe the dissipative bosonic channel and  introduce the Gaussian QFI. We then move to the characterization of the idler-free (or single-mode) strategy, showing a full characterization for the zero and finite temperature cases (Section~\hyperref[sec:III]{III}). In Section~\hyperref[sec:IV]{IV}, we semi-analytically prove that TMSV states are optimal probes for the estimation of the lossy parameter. In Section~\hyperref[sec:V]{V}, we discuss the optimal QFI case, and the relevance of the environment normalization for the QFI. In Section~\hyperref[sec:VI]{VI}, we discuss the implication of our results for quantum hypothesis testing, focusing particularly on the quantum illumination and quantum reading protocols.

\begin{figure}
    \centering
    \includegraphics{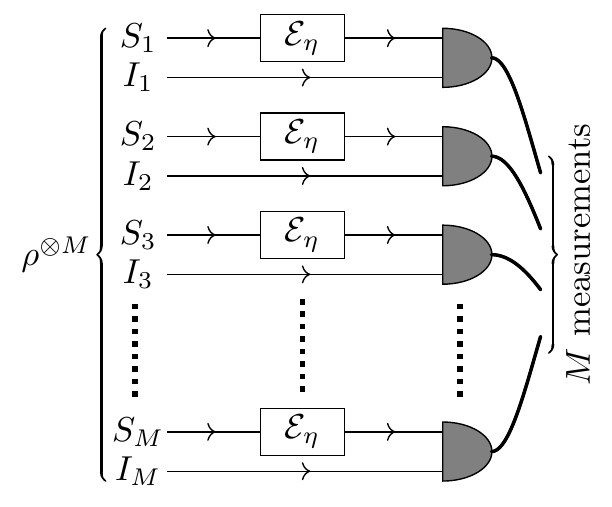}
    \caption{Setup of $M$  probes, each consisting of a Signal and Idler pair, that are used to interrogate the channel $\mathcal{E}_\eta=\mathrm{e}^{-2\ln(\eta)\mathcal{L}}$, where $\mathcal{L}[\rho]=(1+N_B)\mathcal{D}(a_S)[\rho]+N_B\mathcal{D}(a_S^\dag)[\rho]$. Each use of the channel is measured independently and then an estimate of $\eta$ is declared for the collection of results.}
    \label{fig:overview}
\end{figure}

\section{Setup and Methods}\label{sec:II}

\subsection{The lossy bosonic channel}

We consider the bosonic dissipative channel described by the Lindblad generator 
\begin{align}\label{Lindblad_lossy}
\partial_t\rho=\gamma(1+N_B)\mathcal{D}(a)[\rho]+\gamma N_B\mathcal{D}(a^\dag)[\rho],
\end{align}
where $\mathcal{D}(L)[\cdot]=L\cdot L^\dag-\frac{1}{2}\{L^\dag L,\cdot\}$, and $\gamma,N_B\geq0$ are parameters describing the coupling with the bath and the number of noise photons, respectively. This dynamics can be seen in the Heisenberg picture as an attenuation channel, {\it i.e.},
\begin{align}\label{dyn}
    a_S(t)= \eta(t)\,a_S+\sqrt{1-\eta^2(t)}\,h,
\end{align}
where $\eta(t)=\mathrm{e}^{-\gamma t/2}$ is the lossy transmission and $h$ is a thermal mode with $\langle h^\dag h\rangle=N_B$. In the following, we denote the input signal power as $\langle a^\dag a\rangle=N_S$. The channels in Eqs.~\eqref{Lindblad_lossy}-\eqref{dyn} is clearly Gaussian-preserving, as the input-output relation in Eq.~\eqref{dyn} is linear in $a_S$ and $a_S^\dag$. Therefore, the first and second moments of $a_S(t)$ fully characterize the dynamics. In the following, we focus on the value of $\eta(t)$ for a fixed time $t=\bar t$. To simplify the notation, we denote $\eta(\bar t)\equiv \eta$.

It is convenient to work in the covariance matrix formalism. Assume an input composed of a single mode signal ($S$) and an idler ($I$), we use the convention of quadratures ${\bf R} = ({q}_S,{p}_S,{q}_I,{p}_I)^\top$ with the commutator relations $[{R}_i,{R}_j] = \mathrm{i}\Omega_{ij}$, where ${\bf \Omega} = \mathbb{I}_2 \otimes  (\mathrm{i}\sigma_y)$ is the symplectic form. In this convention, the elements of the covariance matrix ${\bf \Sigma}$ are $\Sigma_{ij} = \frac{1}{2}\braket{{R}_i {R}_j+{R}_j{R}_i} - \braket{{R}_i}\braket{{R}_j}$, while the elements of the first-moment vector ${\bf d}$ are $d_i=\langle  R_i\rangle$. The covariance matrix  respects the Heisenberg relation, which can be cast as ${\bf \Sigma}+\mathrm{i}{\bf \Omega}/2\succeq 0$~\cite{Serafini:book}.
The generic signal-idler covariance matrix and first moments can be decomposed as 
${\bf \Sigma}=\begin{bmatrix}
{\bf \Sigma}_S & {\bf \Sigma}_{SI} \\ {\bf \Sigma}_{SI}^\top & {\bf \Sigma}_I
\end{bmatrix}$ and ${\bf d}= \left[
{\bf d}_S^\top, {\bf d}_I^\top
\right]^\top$ respectively, where ${\bf \Sigma}_{S,I,SI}$ are $2\times 2$ matrices and ${\bf d}_{S,I}$ are $2$-dimensional vectors. Here, ${\bf \Sigma}_{S}$ and ${\bf \Sigma}_{S}$ are the covariance matrices of the signal and idler modes, respectively, while ${\bf \Sigma}_{SI}$ is their cross-correlations. The output of the channel in Eq.~\eqref{Lindblad_lossy} can be written as 
\begin{align}
 \tilde {\bf d}(\eta) &=  \left[\begin{matrix}
\eta {\bf d}_S \\ {\bf d}_I
\end{matrix}\right]\label{eq:displdyn} \\
\tilde {\bf \Sigma}(\eta) &= \begin{bmatrix}
\eta^2 {\bf \Sigma}_S+y(\eta) \mathbb{I}_2 & \eta {\bf \Sigma}_{SI} \\ \eta {\bf \Sigma}_{SI}^\top & {\bf \Sigma}_I
\end{bmatrix},\label{eq1}
\end{align}
where $y(\eta)=\left(1-\eta^2 \right)\left(N_B+\frac{1}{2}\right)$. Notice that the relation $2y(\eta)\geq |1-\eta^2|$ ensures that the channel is physical. The idler-free case is given by setting ${\bf \Sigma}_{SI}=\begin{bmatrix}
0 & 0 \\ 0& 0
\end{bmatrix}$, which ensures that the signal and the idler are uncorrelated.

\subsection{Gaussian quantum Fisher information}

In the task of estimating the parameter $\eta$ an experimentalist prepares $M$ copies of an idler-signal system, obtaining as output $M$ copies of the state $\rho(\bar t)$. The experimentalist measures an observable $O$. The induced signal-to-noise ratio (SNR) is defined as 
\begin{equation}\label{SNRdef}
S_\eta[O]=\frac{[\partial_\eta \langle O\rangle_\eta]^2}{\Delta O^2_\eta},
\end{equation}
where $\Delta  O^2_\eta=\langle O^2\rangle_\eta-\langle O\rangle_\eta^2$ and the index $\eta$ indicates the expectation value computed on the state $\rho(\bar t)$. The SNR computed at $\eta=\eta_0$ should be interpreted as the precision achievable for estimating the parameter $\eta$ when its value is close to $\eta_0$, through the relation $\Delta \hat \eta^2_{|\eta\simeq \eta_0}\simeq [M \times S_{\eta_0}]^{-1}$, where $\Delta^2\hat \eta$ is the variance of the estimator $\hat \eta$. Generally speaking, if an experimentalist is able measure a set of observables $\{O_i\}$, they would like to maximize the SNR with respect to this set in order to obtain a better precision rate (call $O_{\rm max}$ the maximizing observable). This, in principle, requires the prior knowledge of $\eta_0$. If this knowledge is not provided, then they can implement a two-step adaptive protocol, where first they measure an observable $A\in\{O_i\}$ such that the function $f(\eta)=\langle A\rangle_\eta$ is invertible in the range of values where $\eta$ belongs, obtaining a first order estimation of  $\eta_0$. Then they find and measure $O_{\rm max}$. The ultimate value of the SNR, {\it i.e.}, $I_\eta = \max_{O} S_\eta[O]$, is the QFI. As already mentioned, the QFI is related to the achievable uncertainty by an unbiased estimator $\hat \eta$ of the parameter $\eta$ via the Cramér-Rao bound: $\Delta^2\hat \eta\geq (\mathcal{I}_\eta)^{-1}$, where  $\mathcal{I}_\eta\equiv MI_\eta$ is the total QFI. 

Since the output state of the channel is Gaussian, it can be represented by the covariance matrix $\tilde {\bf \Sigma}(\eta)$ and the first-moment vector $\tilde {\bf d}(\eta)$. The QFI on this Gaussian manifold is given by~\cite{Safranek2017,Serafini:book}
\begin{equation}
	I_\eta = \mathrm{Tr}\left\{ {\bf L}_2 
\partial_\eta 
	\tilde {\bf \Sigma} \right\} + (\partial_\eta
	\tilde {\bf d})^\top
	\tilde{\bf \Sigma}^{-1}(\partial_\eta \tilde {\bf d}),\label{eq:QFIgeneral}
\end{equation}
where $\mathbf{L}_2$ is the quadratic form of the symmetric logarithmic derivative (SLD), and 	$\tilde{\bf \Sigma}^{-1}(\eta)$ is the pseudoinverse of $	\tilde{\bf \Sigma}(\eta)$. The SLD is the solution to the equation $4 \tilde {\bf \Sigma}{\bf L}_2\tilde {\bf \Sigma} + 
	{\bf \Omega}{\bf L}_2{\bf \Omega}   
= 2\,	\partial_\eta \tilde {\bf \Sigma}$. 
In the following, to simplify the notation, we will simply drop the $\eta$ dependence of the covariance matrix and the first-moment vector.

In the idler-free protocol, QFI can alternatively be expressed, for the  the single-mode case~\cite{,Serafini:book}, as
\begin{equation}\label{QFIsingle}
I_\eta = \frac{ \text{Tr} \left\{ \left( \tilde {\bf \Sigma}^{-1} \partial_\eta \tilde {\bf \Sigma} \right)^2 \right\}}{2 (1+\mu^2)}\ +\frac{2 (\partial_\eta \mu)^2 }{1-\mu^4}+ (\partial_\eta \tilde{\bf d})^\top\tilde{\bf \Sigma}^{-1}(\partial_\eta \tilde{\bf d}),
\end{equation}
where $\mu(\eta) = \left[4 \det \tilde{\bf\Sigma}(\eta)\right]^{-1/2}$ is the purity of the single-mode quantum state. 

Since we are considering the estimation of a parameter embedded in a completely positive and trace preserving map, the QFI is convex~\cite{Fujiwara2001}, and therefore maximized by a pure-state input. We will then consider pure-state for both the idler-free and and entanglement-assisted strategies. Finally, we notice that the QFI of $\eta$ can be used to compute the ultimate precision limit for the estimation of $\gamma$ via the relation $I_\gamma (\gamma)=\frac{t^2}{4}\mathrm{e}^{-\gamma t}I_\eta(\eta = \mathrm{e}^{-\gamma t/2})$.

In the following, we will denote  $I_\eta$, $I_\eta^{\rm IF}$ and $I_\eta^{\rm EA}$ as the QFIs for a generic multi-mode, single-mode and two-mode states, respectively. We will denote the zero temperature case ($N_B=0$) with the suffix ``(0)''. For instance, $I_\eta^{\rm IF, (0)}$ is the generic idler-free (or single-mode) QFI for $N_B=0$.

\section{idler-free protocol}\label{sec:III}

In this section, we discuss the performance of the idler-free (or single-mode) protocol. Part of the discussion is a review of some of the results of Refs.~\cite{Monras2007,Monras2011,Nair2018} with our notations. We separately discuss  the $N_B=0$ and $N_B>0$ cases. Our novel results consist in a characterization of the optimal probe for finite and infinite $N_S$. In particular:
\begin{itemize}
    \item In the $N_B=0$ case, we characterize the transition between the squeezed-vacuum state and a displaced squeezed state as optimal probe. In addition, we provide the conditions for the coherent states to be the optimal probe. 
    \item In the $N_B>0$ case, we characterize an additional transition of the optimal probe happening for sufficiently low $N_S$: from squeezed-vacuum to coherent state. We show that, similarly to the $N_B=0$ case, a displaced squeezed state with an infinitesimal squeezing is the optimal probe in the asymptotic regime ($N_S\gg1$). We also provide the scaling of the optimal squeezing, generalizing the result of Ref.~\cite{Monras2007} to generic temperatures. 
    \item We compute how the simple homodyne detection performs for generic parameter values, showing that it does not realize the $(1-\eta^2)^{-1}$-scaling of the optimal QFI. This means that photon counting is needed to achieve the optimal precision in the $1-\eta\ll1$ regime.
\end{itemize}

\subsection{Parametrization}
In the idler-free protocol, $M$ independent copies of a {\it single}-mode state are sent as input of the channel. A generic Gaussian single-mode state can be parametrized as
\begin{align}\label{single}
    \mathbf{d}_S &= \begin{bmatrix}
     q \\ p
    \end{bmatrix},\\
    {\bf \Sigma}_S &= \begin{bmatrix}
     ar & 0 \\ 0 & ar^{-1}
    \end{bmatrix}.
\end{align}
Here, $a\geq 1/2$ and $r>0$ ensure that the state is physical: $r=1$ means no squeezing, while $r\to0$ ($r\to\infty$) corresponds to infinite squeezing (amplification). Since the QFI is convex, it is maximized for a pure input-state~\cite{Fujiwara2001}. Therefore, we set $a=1/2$, where only squeezing and displacement play a role. 

Let us denote the total number of signal photons by $N_S=N_{\rm coh}+N_{\rm sq}$, where $N_{\rm coh}=(p^2+q^2)/2$ is the displacement contribution, and $N_{\rm sq}=(r+r^{-1}-2)/4$ is the squeezing contribution. The quadratures can be parametrized as $q=\sqrt{2N_{\rm coh}}\cos\theta$ and $p=\sqrt{2N_{\rm coh}}\sin\theta$. Moreover, we have that $r=1+2N_{\rm sq}-2\sqrt{N_{\rm sq}\left(N_{\rm sq}+1\right)}$, where we have imposed that $r\in(0,1]$. This allows to write the QFI in terms of $N_{\rm sq}$ and $N_{\rm coh}$. 
The general estimation strategy consists in using a properly optimized displaced squeezed state as probe. Therefore, as a further step, we consider the parametrization defined by $N_{\rm sq} = \xi N_S$ and $N_{\rm coh} = N_S\left(1-\xi\right)$, where $\xi\in [0,1]$ is the ratio of squeezed photons to the total number of signal photons. We will denote as $\xi^{\rm opt}$ the ratio optimizing the QFI. 

The idler-free QFI $I_\eta^{\rm IF}$ can be now computed using Eq.~\eqref{QFIsingle} and evaluated with a symbolic computation software. The following Lemma notably simplifies the analysis.
\begin{lemma}
The displacement angle optimizing the single-mode QFI for any parameter values is $\theta = n\pi$, with $n\in\mathbb{N}$. 
\end{lemma}
\begin{proof}
We have that
\begin{equation}
    I_\eta^{\rm IF}(\theta=n\pi) - I_\eta^{\rm IF} = \frac{4\eta^2 \left(1-r^2\right)}{\left(\eta^2+2ry\right)\left(r\eta^2+2y\right)}N_{\rm coh}\sin^2(\theta),
    \end{equation}
    where 
    $y=(1-\eta^2)(N_B +1/2)$. This quantity is non-negative for any parameter values and is zero for $\theta = n\pi$.
\end{proof}

In the following, we consider solely probes displaced along the optimized angle $\theta_{\rm opt}=n\pi$, and denote for simplicity $I_\eta^{\rm IF}\equiv I_\eta^{\rm IF}(\theta_{\rm opt})$.
Notice that even if finding the optimal probe for a given channel in the energy-constrained case is now brought to a one-variable optimization problem, it remains still an highly parametrized problem. Understanding the relevant asymptotic regimes is crucial to fully characterize the QFI.

\subsection{The zero temperature case: $N_B=0$}
This case has been studied in Refs.~\cite{Monras2007,Monras2011} in the Gaussian case. Here, we derive novel analytical results for the optimal states in the energy-constrained case. In this case, the QFI takes a relatively simple form: 
\begin{align}\label{QFI_NB0}
    I_\eta^{\rm IF, (0)}&=4N_S\left\{\frac{1-\xi}{1-2\eta^2\left(\sqrt{\xi N_S(1+\xi N_S)}-\xi N_S\right)}
+\frac{\xi\left[(1-\eta^2)^2+\eta^4\right]}{(1-\eta^2)(1+2\xi N_S\eta^2(1-\eta^2))}\right\}.
\end{align}
Our task consists in finding $\xi$ that optimizes $I_\eta^{\rm IF, (0)}$ for given values of $N_S$ and $\eta$. This problem can be solved numerically for arbitrary parameter values, see Fig.~\ref{fig:xi_surface}. However, we seek to find the analytical behaviour of the optimal probe.
Let us first state a simple bound on the QFI, that will be useful in the discussion.
\begin{lemma}~{\bf \cite{Nair2018}}\label{bound}
The  QFI for a generic multi-mode probe is bounded by $I_\eta^{\rm (0)}\leq \frac{4N_S}{1-\eta^2}$ for any $\eta\in[0,1)$.
\end{lemma}

Generally speaking, both displacement and squeezing are essential for achieving optimality. However, it is interesting to look for the regimes where squeezing or displacement alone are the optimal probes. In Fig.~\ref{fig:xi_surface} we can see a transition between $\xi^{\rm opt}=1$ and $\xi^{\rm opt}<1$. The following proposition characterizes this transition.

\begin{prop}{\bf [Squeezed-vacuum state as optimal probe ($N_B=0$)]}\label{propNS}
$\xi^{\rm opt}=1$ if and only if $N_S\leq \bar N_S(\eta)$. Here, 
\begin{align}
    \bar N_S(\eta)= \left\{ \begin{array}{ll}
            0 & \quad \eta\leq 1/\sqrt{2} \\
            \bar N_S^{(0)}(\eta) & \quad \eta>1/\sqrt{2},
        \end{array} \right.
\end{align}
where $\bar N_S^{(0)}(\eta)$ is the only zero of $f_1(\eta,N_S)=\frac{1}{4N_S}\left(\partial_\xi I_\eta^{\rm IF, (0)}\right)_{|\,\xi=1}$.
\end{prop}
\begin{proof}
In Appendix~\hyperref[app:A2]{A2} we show that $I_\eta^{\rm IF, (0)}$ is concave in $\xi$, provided that $\eta\not=0$, see Appendix~\hyperref[app:A2]{A2}. This means that $\xi=1$ is a maximum point if and only if  $f_1(\eta,N_S):=\frac{1}{4N_S}\left(\partial_\xi I_\eta^{\rm IF, (0)}\right)_{|\,\xi=1}\geq0$. 
We have that $f_1$ has at most one zero, as its derivative in $N_S$ is negative everywhere, see Appendix~\hyperref[app:A2]{A2}. Since $f_1\to-\frac{1}{1-\eta^2}<0$ for $N_S\to \infty$, we have that the zero $\bar N_S^{(0)}(\eta)$ is positive only if $f_1(\eta,N_S=0)=\frac{2\eta^2-1}{1-\eta^2}$, is positive, \textit{i.e.}, for some $\eta\in\left(\frac{1}{\sqrt{2}},1\right)$.
\end{proof}
Proposition~\ref{propNS} implies that the squeezed-vacuum state is the never optimal for $\eta\leq\frac{1}{\sqrt{2}}$, or if the input power $N_S$ is large enough. More precisely, the squeezed-vacuum state is optimal only for $\eta \geq \bar \eta(N_S)$, where $\bar \eta(N_S)$ is the inverse of $\bar N_S^{(0)}(\eta)$. The curve defined by $f_1=0$ can be computed numerically, and an analytical expansion can be derived using perturbation theory. For instance, a perturbation expansion to the first order gives us $\bar \eta\simeq 1-\frac{1}{cN_S}$, with $c\simeq8.86$, for $N_S\gg1$, and $\bar\eta\simeq \frac{1}{\sqrt{2}}\left(1+\frac{\sqrt{N_S}}{2}\right)$ for $N_S\ll1$, see Appendix~\hyperref[app:A3]{A3}.

\begin{figure}[t!]
    \centering
    \includegraphics[width=0.97\linewidth]{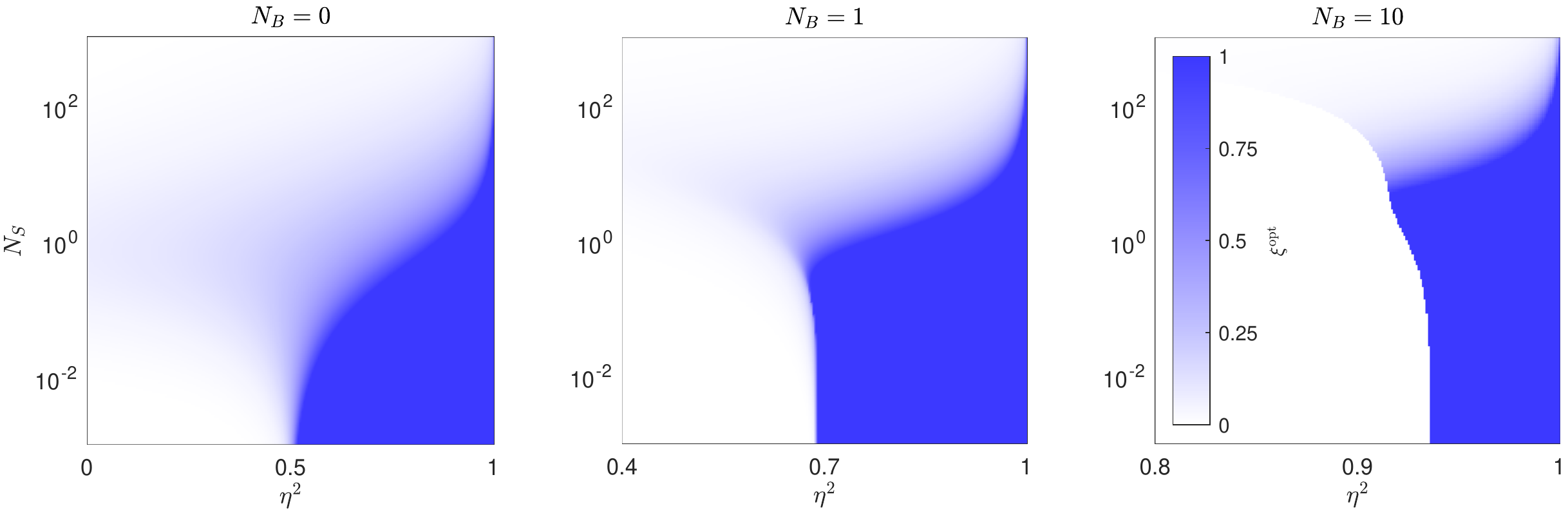}
      \caption{The optimal ratio of squeezed photons as a function of number of signal power ($N_S$) and lossy transmission ($\eta$), for three cases of background noise. At low power there is a sharp transition from the coherent state ($\xi^{\rm opt}=0$) to the squeezed-vacuum being optimal ($\xi^{\rm opt}=1$). In the moderate power regime, there is a region where a non-trivial displaced squeezed state is optimal. In the high power regime, we have that an infinitesimal squeezing is necessary for ensuring optimality ($\xi^{\rm opt}\to0$).}
    \label{fig:xi_surface}
\end{figure}

\begin{figure}[t!]
    \centering
    \includegraphics[width=\linewidth]{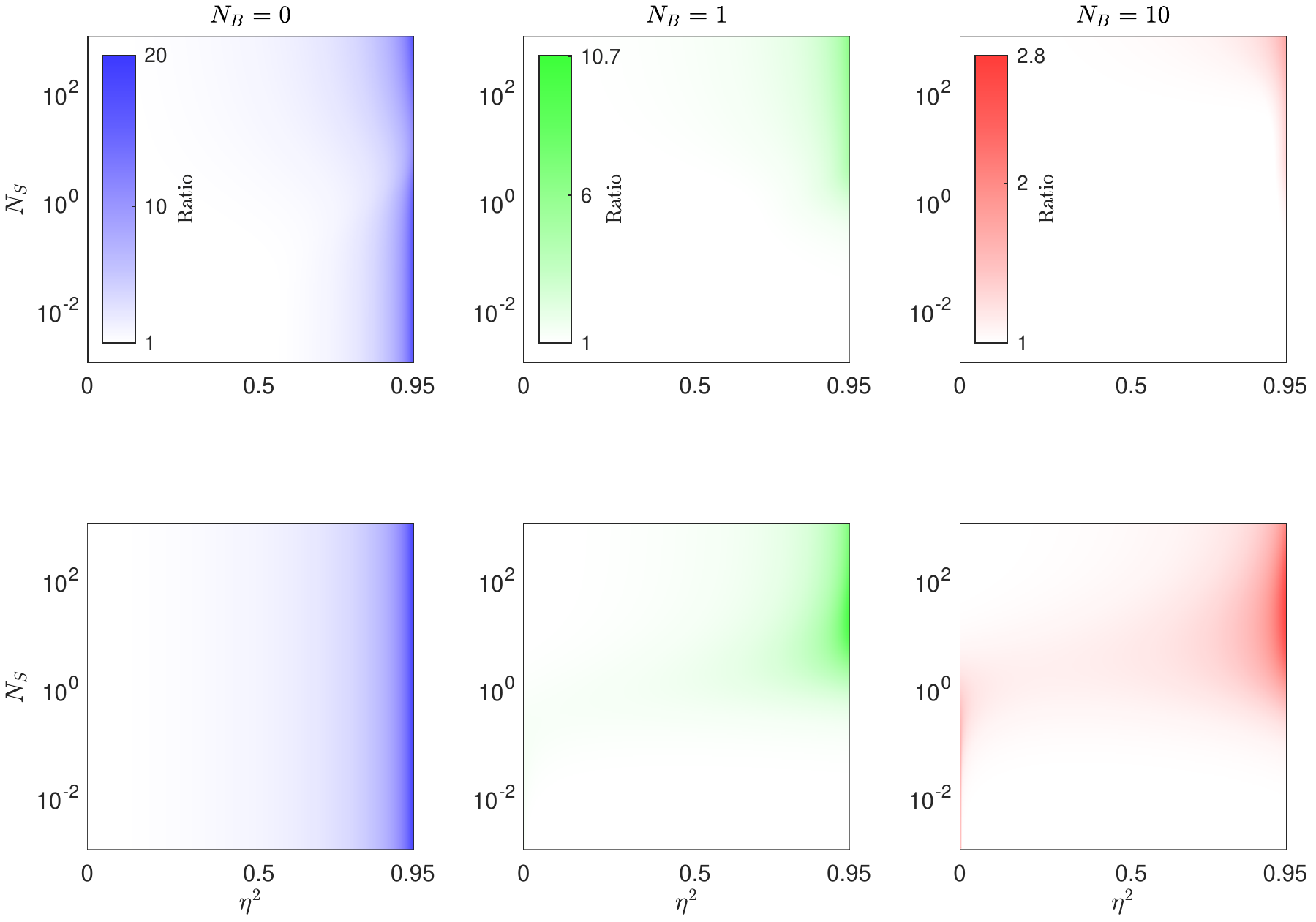}
    \caption{Quantum advantage as the ratio between optimized QFI and coherent state QFI. Colour scaling is shared within each column. \textbf{Top}: Optimal idler-free case. \textbf{Bottom}: Optimal entanglement-assisted case, achieved by a TMSV state.}
    \label{fig4}
\end{figure}

Understanding whether coherent states performs optimally in certain regimes is important, as these states are a close representation of a classical signal. Due to this property, many sensing protocols are compared with respect to coherent states in order to claim a quantum advantage, see Refs.~\cite{Tan2008, Pirandola2011} among others. The following is a no-go result for the coherent state as optimal probe.

\begin{prop}{\bf [Coherent state as optimal probe ($N_B=0$)]}\label{propCOH}
The coherent state ($\xi=0$) cannot be the optimal probe for any $\eta>0$.
\end{prop}
\begin{proof}
Due to the concavity of $I_\eta^{\rm IF, (0)}$ for $\eta\not=0$, the coherent state is optimal if and only if $\frac{1}{4N_S}\left(\partial_\xi I_\eta^{\rm IF, (0)}\right)_{|\,\xi=0}\leq0$. However, we have that $\frac{1}{4N_S}\left(\partial_\xi I_\eta^{\rm IF, (0)}\right)_{|\,\xi=0}=\frac{\sqrt{N_S}\eta^2}{\sqrt{\xi}}+\mathcal{O}(1)$ for $\xi\to0$, which is positive for any $\eta>0$.
\end{proof}
Let us indeed investigate the $\eta\to0$ limit, and show that there are regimes where coherent state is not optimal even for in this regime. We have that 
\begin{equation}\label{eq:g1}
    I_\eta^{\rm IF, (0)}=
 4N_S \left\{1+g_1(\xi,N_S)\eta^2\right\}+\mathcal{O}(\eta^4), \quad \mathrm{as}\  \eta\to0,
\end{equation} 
where $g_1(\xi,N_S)=2(1-\xi)\sqrt{\xi N_S(1+\xi N_S)}-\xi  (1+2N_S)$. In the regimes $N_S\ll1$ and $N_S\gg1$, the function $g_1(\xi,N_S)$ is always negative and decreasing with respect to $\xi$. This implies that the coherent state ($\xi=0$) is optimal in these limits. However, for intermediate values of $N_S$,  the function $g_1(\xi,N_S)$ is positive for some finite $\xi$, meaning that the QFI is maximized for a displaced squeezed state. This behaviour of the QFI is clearly visible in Fig.~\ref{fig:xi_surface}.

We now move the discussion to the regimes where non-trivial displaced squeezed states optimize the QFI. In particular, we are interested in the high- and low-power regimes, where some interesting properties emerge. In the large power regime, we have that
\begin{equation}\label{coheNSlarge}
    I_\eta^{\rm IF, (0)}= 4N_S\frac{(1-\xi)}{1-\eta^2}+\mathcal{O}(1), \quad \mathrm{as}\ \xi N_S \to \infty.
\end{equation}
In this limit the optimal squeezing is infinitesimal, {\it i.e.}, $\xi^{\rm opt}\to0$. However, $\xi^{\rm opt}$ cannot be exactly zero, otherwise the $(1-\eta^2)^{-1}$-scaling of the QFI disappears, as one can see using Eq.~\eqref{QFI_NB0}. By expanding Eq.~\eqref{coheNSlarge} to the next order in  $\xi N_S$, we derive the asymptotic value $\xi^{\rm opt}\sim \eta/[4N_S(1-\eta^2)]^{1/2}$, see Appendix~\hyperref[app:A4]{A4}. This asymptotic expansion holds for $N_S\gg \eta^2/(1-\eta^2)$~\cite{Footnote1}. 
Interestingly, this means that in the $N_S\gg1$ regime, an infinitesimal amount of squeezing ensures the optimality of the QFI. Notice also that Eq.~\eqref{coheNSlarge} virtually saturates the bound in Lemma~\ref{bound}. Therefore, the single-mode state is asymptotically an optimum among generic multi-mode states.

In the low-power regime, we have that
\begin{equation}\label{smallN_S}
    I_\eta^{\rm IF, (0)}= 4N_S\left\{(1-\xi)+\frac{\xi\left[(1-\eta^2)^2+\eta^4\right]}{(1-\eta^2)}\right\}+\mathcal{O}(N_S^{3/2}), \quad \mathrm{as}\  N_S\to0.
\end{equation}
This is a linear quantity in $\xi$, meaning that in this limit there is an abrupt change in the optimal $\xi$: $\xi^{\rm opt}=0$ for $\eta<1/\sqrt{2}$, and $\xi^{\rm opt}=1$ otherwise. We will see that this transition is even more evident in the finite temperature case, {\it i.e.}, for $N_B>0$.
Finally, we have that in the intermediate power regime a finite squeezing is {\it always} a resource in the quantum estimation task, even for small $\eta$, as shown in Eq.~\eqref{eq:g1}. This happens especially in the $10^{-1}\lesssim N_S\lesssim 10$ regime. However, Fig.~\ref{fig4} tells us that the advantage is minimal for $\eta\lesssim 1/\sqrt{2}$, and it becomes increasingly relevant only for $\eta$ approaching one.

\subsection{Finite temperature case: $N_B>0$}
\subsubsection{Shadow effect}
In the finite temperature case we have a peculiar feature, which consists in the vacuum having metrological power: 
\begin{align}\label{shadow}
    I_\eta(N_S=0)=\frac{4\eta^2N_B}{\left(1-\eta^2\right)\left[1+N_B\left(1-\eta^2\right)\right]}\equiv I_{\rm shad}.
\end{align}
 This is an effect appearing for $\eta,N_B>0$, consisting of a sort of shadow that the system generates in a non-vacuum environment. It is present for a generic multi-mode state. This feature, indeed, could not be observed in Refs.~\cite{Monras2007, Monras2010}, where the analysis is limited to the  vacuum environment. We call this ``shadow-effect"~\cite{TanPHD}, and denote its contribution to the QFI as $I_{\rm shad}$, as in Eq,~\eqref{shadow}. 

\subsubsection{Coherent and squeezed-vacuum probes}

In order to gain an intuition on the optimal probe, let us first discuss the QFI of two topical states: coherent and squeezed-vacuum states. For a coherent state as input, \textit{i.e.}, for $\xi=0$, the QFI can be written in a closed form as
\begin{equation}\label{coherent}
   I_\eta^{\rm IF}(\xi=0)   =I_{\rm shad} + \frac{4N_S}{1+2N_B\left(1-\eta^2\right)}\equiv I_\eta^{\rm coh}.
\end{equation}
Notice that coherent states asymptotically achieve the Schr\"odinger's precision limit for any parameter values.
For a squeezed-vacuum state probe, \textit{i.e.}, for $\xi=1$, we have a lengthy expression for the QFI, that we denote as $I_\eta^{\rm IF}(\xi=1)\equiv I_\eta^{\rm sq}$, see Appendix~\hyperref[app:A1]{A1}. In Fig.~\ref{fig:xi_surface}, we see the presence of a clear region where $\xi^{\rm opt}=1$. This feature is similar to what proved in Prop.~\ref{propNS} in the zero temperature case. In the large squeezing regime, the QFI saturates to a $\eta$-dependent value:
\begin{align}\label{asymp:sq0}
    I_{\eta}^{\rm sq}= \frac{2(1-\eta^2)^2+2\eta^4}{\eta^2(1-\eta^2)^2}+\mathcal{O}(N_S^{-1}), \quad\mathrm{as}\  N_S\to \infty.
\end{align}
This limit holds for $N_S\gg (1+N_B)/[\eta^{2}(1-\eta^2)]$. Let us investigate $I_{\eta}^{\rm sq}$ at the diverging points of Eq.~\eqref{asymp:sq0}. The analysis of the different regimes is complicated by the fact that the order of different limits do not commute. However, one can rely on Taylor analysis to understand which limit order corresponds to which regime of parameters, see Appendix~\hyperref[app:A6]{A6} for a discussion on this. 
The limits $\eta\to0$ and $N_S\to\infty$ do not commute, as  $I_{\eta}^{\rm sq}=\mathcal{O}(\eta^2)$ while $\eta=0$ is a diverging point of Eq.~\eqref{asymp:sq0}. This is due to the fact that Eq.~\eqref{asymp:sq0} holds for $N_S\eta^2\gg 1$, while $I_{\eta}^{\rm sq}=\mathcal{O}(\eta^2)$ holds for $N_S\eta^2\ll 1$.
Instead, at $\eta=1$ we have that 
\begin{align}\label{asymp:sq}
    I_{\eta}^{\rm sq}= \frac{2N_S(1+2N_B)+2N_B}{1-\eta}+\mathcal{O}(1), \quad \mathrm{as}\ \eta\to1,
\end{align}
At first glance, this may seem in contrast with Eq.~\eqref{asymp:sq0}, as Eq.~\eqref{asymp:sq} is unbounded with respect to $N_S$. Indeed, as a Taylor analysis  reveals, Eq.~\eqref{asymp:sq0} is valid for $N_S(1-\eta)\gg 1$ while Eq.~\eqref{asymp:sq} holds for  $N_S(1-\eta)\ll 1$. This means that squeezed-vacuum states {\it do not} asymptotically reach the Schr\"odinger's precision limit, as their QFI saturates for large enough $N_S$ for any fixed value of $\eta<1$, \textit{i.e.}, $I_{\eta}^{\rm sq}/N_S\to 0$ for $N_S\to \infty$.  

\subsubsection{Optimal probe}
Let us now consider the general case of a displaced squeezed state probe. In Fig.~\ref{fig:xi_surface}, we see that squeezing can be resource even when $\eta$ is far from being one. In the low-power regime there is an abrupt transition from $\xi^{\rm opt}=0$ to $\xi^{\rm opt}=1$ at a certain value of $\eta$. This can be seen more clearly by expanding $I_\eta$ for small $N_S$:
\begin{align}
    I_\eta^{\rm IF}=I_{\rm shad}+4N_S\left\{\frac{1-\xi}{1+2N_B(1-\eta^2)}+\xi g_2(\eta,N_B)\right\}+\mathcal{O}(N_S^{3/2}),\quad\mathrm{as}\  N_S\to0
\end{align}
where $g_2(\eta,N_B)$ is given in Appendix~\hyperref[app:A5]{A5}. It is clear that $\xi^{\rm opt}=1$ if $g_2(\eta,N_B)>\frac{1}{1+2N_B(1-\eta^2)}$, otherwise $\xi^{\rm opt}=0$. For large $N_B$, we have that the abrupt change happens at $\eta\simeq 1-\frac{3}{2N_B}$, see Appendix~\hyperref[app:A5]{A5}.

In the large power regime, $\xi^{\rm opt}$ behaves similarly as in the $N_B=0$ case, as shown in Fig.~\ref{fig:xi_surface}. More precisely, we have the following result for the asymptotic QFI, which generalizes (and include) the $N_B=0$ case.
\begin{prop}{\bf [Optimal asymptotic QFI]}
The optimal QFI in the large power regime is given by
\begin{align}\label{sq:NSlarge}
    I_\eta^{\rm IF}=\frac{4N_S(1-\xi)}{(1-\eta^2)(1+2N_B)}+\mathcal{O}(1), \quad\mathrm{as}\  \xi N_S\to \infty.
\end{align}
Here, the optimal squeezing is given by $\xi^{\rm opt}\sim \eta/[4N_S(1-\eta^2)(1+2N_B)]^{1/2}$ for $ N_S\gg\eta^2/[(1-\eta^2)(1+2N_B)]$ and $N_S\gg N_B$.
\end{prop} 
\begin{proof}
The Taylor expansion for large $\xi N_S$ is
\begin{align}
    I_\eta^{\rm IF}=\frac{4N_S(1-\xi)}{(1-\eta^2)(1+2N_B)}\left[1-\frac{\eta^2}{4N_S\xi(1-\eta^2)(1+2N_B)}\right]+\frac{2(1-2\eta^2+2\eta^4)}{\eta^2(1-\eta^2)^2}+{ \scriptstyle \mathcal{O}}(1),
\end{align}
which holds for $\xi N_S\gg\eta^2/[(1-\eta^2)(1+2N_B)]$ and $N_S\gg N_B$. By setting the derivative with respect to $\xi$ to zero and solving for $\xi$, we obtain $\xi^{\rm opt}\sim \eta/[4N_S(1-\eta^2)(1+2N_B)]^{1/2}$.
\end{proof}

\subsection{Homodyne detection}

To realize the full benefits in using an optimized probe, the receiver must be optimized accordingly, in order for  the classical Fisher information to saturate the QFI. For Gaussian probes, the optimal receiver includes up to quadratic terms. Generally, this can be implemented by a linear circuit and photon counting. It is of experimental interest to understand what performance a simple detection scheme, such as homodyne, can achieve. Let us compute the classical Fisher information for homodyne detection on the probe optimizing the QFI. If $Q_x$ is a Gaussian random variable parametrized by a scalar unknown $x$, \textit{i.e.}, $Q_x\sim\mathcal{N}\left(m(x),V(x)\right)$, then the Fisher information of $x$ due to $Q_x$ is $H_x=\left(\partial_xm\right)^2V^{-1}+2^{-1}\left(\partial_xV\right)^2V^{-2}$.  We have that, for the probe state with $\mathbf{d}=\left(\sqrt{2N_{\rm coh}},0\right)^T$ and ${\bf\Sigma}_S=2^{-1}{\rm diag}\left(r,r^{-1}\right)$ passed through the channel and measured by homodyne detection along the in-phase quadrature, the Fisher information is
\begin{equation}
    H_\eta = \frac{2\eta^2\left(1+2N_B-r\right)^2}{\left(\eta^2r+ \left(1-\eta^2\right)\left(1+2N_B\right) \right)^2} + \frac{4N_{\rm coh}}{\eta^2r+\left(1-\eta^2\right)\left(1+2N_B\right)}.
\end{equation}
Clearly, homodyne detection is ideal for $\eta^2\ll 1$, since $\frac{H_\eta}{I_\eta^{\rm IF}}\simeq1$. Similarly, homodyne detection does well for strong signals with finite displacement, as for $N_S\gg \left(N_B+\frac{1}{2}\right)[\eta^2\left(1-\xi\right)]^{-1}$, we find $\frac{H_\eta}{I_\eta^{\rm IF}}\simeq1$. Furthermore, the loss due to homodyne detection is only a factor of two in the noisy regime, with $\frac{H_\eta}{I_\eta^{\rm IF}}\simeq\frac{1}{2}$ for $N_B\gg N_S[\eta^2(1-\eta^2)]^{-1}$.  Otherwise, homodyne detection is generally non-ideal. In particular, $H_\eta$ does not realize the $\left(1-\eta^2\right)^{-1}$-scaling, as $\lim_{\eta\rightarrow1}\frac{H_\eta}{I_\eta^{\rm IF}}=0$. In this regime for $\eta$, photon counting is needed to achieve the optimal precision.

\section{Entanglement-assisted strategy}\label{sec:IV}

In this section, we analyse the benefits of having access to an ancilla system, including entanglement. We aim to find the two-mode state that optimizes the QFI. This turns to be a highly parametrized problem, as a Gaussian system has $14$ parameters that can be varied. Here, the method used in Ref.~\cite{Nair2020} to find an ultimate bound on the QFI does not work, as the authors rely strongly on the noise normalization $N_B\rightarrow N_B/(1-\eta^2)$. Indeed, with this normalization, the channel can be represented as a composition of a lossy channel and a $\eta$-independent amplifier channel. This allows to reduce the problem to the zero temperature case, that has been solved in Ref.~\cite{Nair2018}. Without normalization, there is not such decomposition, leaving the $N_B>0$ case unsolved. 

In the following, we first strive to lower the complexity of the problem, by finding the canonical form of the generic pure-state probe. We then optimize the pure-state probe with respect to the displacement angle in a manner similar to the single-mode probe. Finally, we impose the energy constraint to arrive at a two-dimensional optimization problem. This allows us to numerically solve  the problem, and find that TMSV states are optimal for any parameter choice. We further support this result analytically in some special regimes.

\subsection{Parametrization}

Our starting point is the following Lemma, which helps in significantly reducing  the complexity of the problem.

\begin{lemma}{\bf [Canonical form of generic pure-state probe]}\label{prop:covariance}
The covariance matrix for the generic two-mode pure input state of the entanglement-assisted protocol can be written as $\begin{bmatrix}
  {\bf \Sigma}_S & {\bf \Sigma}_{SI} \\
  {\bf \Sigma}_{SI}^\top & {\bf \Sigma}_{I}
\end{bmatrix}$, where
\begin{equation}
    {\bf \Sigma}_S ={\rm diag}\left(ar,ar^{-1}\right),\quad  {\bf \Sigma}_I = {\rm diag}\left(a,a\right), \quad {\bf \Sigma}_{SI} = \sqrt{a^2-\frac{1}{4}}\begin{bmatrix}\sqrt{r}\cos\phi & \sqrt{r}\sin\phi  \\ \sqrt{r^{-1}}\sin\phi & -\sqrt{r^{-1}}\cos\phi\end{bmatrix}. \label{eq:purecovariance}
\end{equation}
\end{lemma}
The proof is given in Appendix~\hyperref[app:B1]{B1}. If we consider also the displacement, this reduces the QFI to a five-parameters quantity. The problem of optimizing the QFI can be further simplified to a two-dimensional problem, by setting the optimal displacement angle and the energy constraint.

\subsubsection{Displacement angle optimization}
We calculate the two-mode QFI for the probe state with covariance matrix as in Eq.~\eqref{eq:purecovariance} and displacement $\mathbf{d} = \sqrt{2N_{\rm coh}}(\cos\theta,\sin\theta,0,0)^T$. By simplification with symbolic software, we verify that the resulting QFI is  independent of the rotation by $\phi$. See Appendix~\hyperref[app:B2]{B2} for the full expression.  Moreover, we have the following Lemma on the optimal displacement angle. 
\begin{lemma}\label{prop:displacement}
The two-mode QFI is maximised for displacement along $\theta=n\pi$,  with $n\in\mathbb{N}$.
\end{lemma}
\begin{proof}
Displacement appears only in the second term of Eq.~\eqref{eq:QFIgeneral}, which is computed, for the covariance matrix probe of Eq.~\eqref{eq:purecovariance} and dynamics as in Eqs.~\eqref{eq:displdyn}--\eqref{eq1}, as
\begin{equation}
    (\partial_\eta \tilde{\mathbf{d}})^\top {\bf \tilde \Sigma}^{-1}(\partial_\eta \tilde{\mathbf{d}}) = 8N_{\rm coh}a\left(\frac{\cos^2\theta}{4ay+r\eta^2}+\frac{\sin^2\theta}{4ay+\eta^2/r}\right),\label{eq:displacement}
\end{equation}
where $y=\left(1-\eta^2\right)\left(N_B+\frac{1}{2}\right)$. If $r=1$, $\theta$ is degenerate. Otherwise, if $r< 1$, Eq.~\eqref{eq:displacement} is maximised for $\theta=n\pi$.
\end{proof}

\subsubsection{Energy constraint}
With optimal displacement along $\theta=n\pi$, the task of optimization is reduced to three parameters. Equivalently to the single-mode optimization, we restrict the total number of photons per mode as  $N_S = N_{\rm coh} +N_{\rm sq.th.}$. We  introduce the free parameter $\zeta^2\in [0,1]$ as the fraction of photons allocated to the covariance. In particular, $N_{\rm coh}=N_S\left(1-\zeta^2\right)$ and $N_{\rm sq.th.}= N_S\zeta^2$. 

The number of photons of a squeezed thermal state with covariance matrix ${\bf \Sigma}_S$ as in Eq.~\eqref{eq:purecovariance} is $N_{\rm sq.th.} = \frac{a}{2}\left(r+r^{-1}\right)  - \frac{1}{2}$. Notice that if we for the moment  fix $\zeta$, we have fixed also the photons allocated to the covariance as $N_{\rm sq.th.}=N_S\zeta^2$.
We use this to eliminate the parameter $a$, as $a=\frac{2N_S\zeta^2+1}{r+r^{-1}}$, and retain the free parameter $r$ which represents the trade-off between local squeezing and correlations.  Since the number of photons allocated to the covariance matrix depends on $\zeta$, so does also the range of possible squeezing, as $r\in [2N_S\zeta^2+1-2\sqrt{N_S\zeta^2\left(N_S\zeta^2+1\right)},1]$. In summary, the energy-constrained two-mode QFI is parametrized  on the two-dimensional  space $(\zeta,r)$.

\subsection{TMSV state as optimal probe}

\begin{figure}
    \centering
    \includegraphics[width=0.9\linewidth]{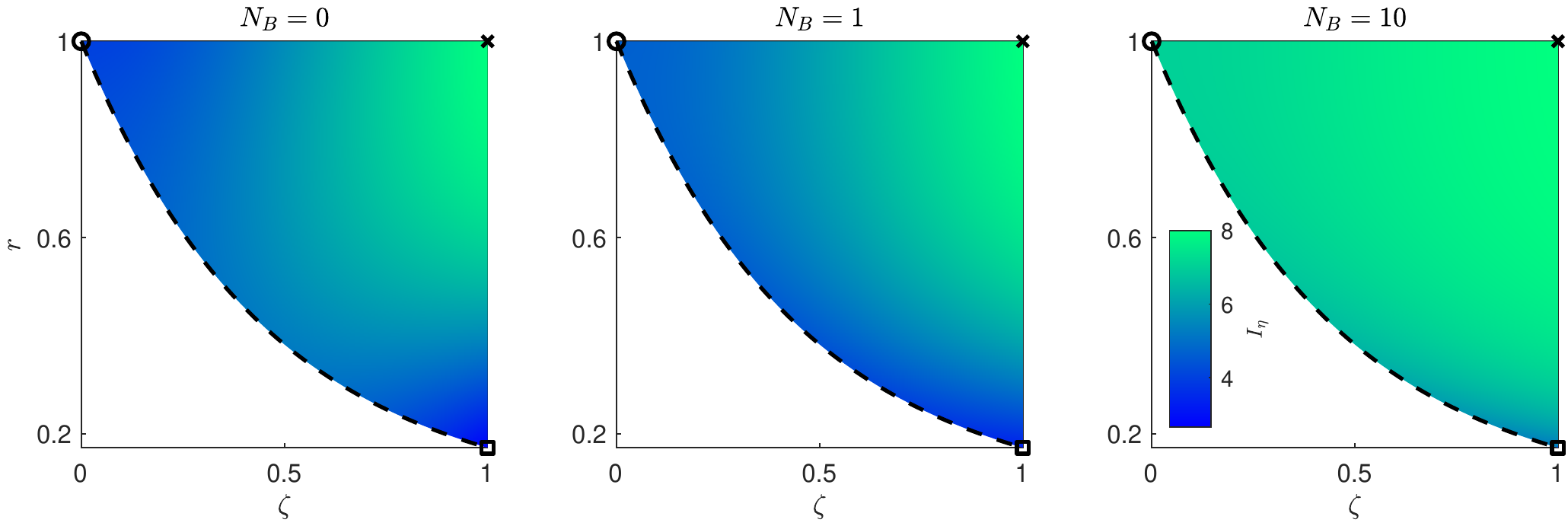}
    \caption{QFI for $\eta=\frac{1}{\sqrt{2}}$ and $N_S=1$ on the parameter space of ($\zeta,r$). The circle, square, and cross indicate coherent state, single-mode squeezed-vacuum state, and two-mode squeezed-vacuum state, respectively. The dashed line indicates the squeezed and displaced single-mode state considered in Fig.~\ref{fig:xi_surface}. For any fixed set of $\{N_S,N_B,\eta\}$, the point $(1,1)$ is maximum. }
    \label{fig:figure4}
\end{figure}

\subsubsection{Numerical results}

We have run exhaustive searches on the two-dimensional parameter space $(\zeta,r)$ to find the point maximizing the two-mode QFI. For each scenario in $\left\{N_S\in\left[10^{-3},10^3\right], N_B\in\left[10^{-3},10^3\right],\eta\in\left[10^{-3},0.999\right]\right\}$, the point $(\zeta=1,r=1)$ always results  to be the global maximum. That is, the optimal strategy always consists of allocating all photons to maximize correlations in the covariance matrix. Indeed, the state corresponding to $(1,1)$ is the TMSV. See also  Fig.~\ref{fig:figure4} for three samples of this verification with varying amounts of background noise. 

\subsubsection{Analytical results}

We support the numerical results analytically by showing that the point $(\zeta=1,r=1)$ corresponds to a local maximum of the QFI.

\begin{prop}{\bf [TMSV as local maximum of the QFI]}\label{prop:maxQFI}
On the parameter space of $(\zeta,r)$, the two-mode QFI is maximized at the point $(1,1)$.
\end{prop}
\begin{proof} The proof consists of evaluating the gradients at the point of interest. Assume a non-zero signal $N_S>0$. We have that $\left( \partial_r I_\eta^{\rm EA} \right)|_{\zeta = 1,r=1} = 0 $, \textit{i.e.} $\left(1,1\right)$ is a stationary point with respect to $r$. Furthermore
\begin{align}
    \left(\partial_r^2I_\eta^{\rm EA}\right) |_{\zeta=1,r=1} = \quad & -\frac{2\left(1+2N_S\right)^2}{g\left(N_S,N_B\right)\left(1+ \left(1-\eta^2\right)g\left(N_S,N_B\right)\right)^2} \cdot  \frac{f_{1,{\rm num}}\left(N_S,N_B,\eta\right)}{f_{1,{\rm den}}\left(N_S,N_B,\eta\right)} < 0.
\end{align}
Here, 
\begin{align}
    f_{1,{\rm num}}= &\quad2\left(1-\eta^2\right)\left[N_S^3\left(1+2N_B\right)\eta^2\left(1-\eta^2\right)\left(2N_B^2+2N_B+1\right) + \left(1+\eta^2\right)^2N_B^4\left(1+2N_S\right) \right]\nonumber \\
    & + 2N_S^2\left(1-\eta^2\right) \left[ 2\left(1+\eta^2\right)^2N_B^4+2\left(2+7\eta^2-\eta^4\right)N_B^3+\left(3+16\eta^2-\eta^4\right)N_B^2 +\left(1+9\eta^2\right)N_B+2\eta^2\right] \nonumber \\
    & + 2N_S\left[2\left(2-\eta^2\right)\left(2-\left(1-\eta^2\right)^2\right)N_B^3+\left(3+7\eta^2-5\eta^4+\eta^6\right)N_B^2+\left(1+4\eta^2-\eta^4\right)N_B+\eta^2\right] \nonumber \\
    & + N_B\left[4\left(1+\eta^2-\eta^4\right)N_B^2+\left(3+3\eta^2-2\eta^4\right)N_B+\eta^2+1\right] > 0, \\
    f_{1,{\rm den}}= \quad & 2\left(1-\eta^2\right)^2\left(2N_S^2N_B^2+2N_S^2N_B+N_S^2+2N_SN_B^2+N_B^2\right)+2g(N_S,N_B)\left(1-\eta^2\right)+1 >0,
\end{align}
where $g(x,y)=x+2xy+y$. That is, the second order derivative is strictly negative. Therefore, the point $(1,1)$ is a maximum with respect to $r$ for any configuration of $\{N_S,N_B,\eta\}$. Regarding the parameter $\zeta$, we have
\begin{equation}
  \left( \partial_\zeta I_\eta^{\rm EA} \right)|_{r=1,\zeta = 1} 
= \frac{2N_Sf_2\left(N_S,N_B,\eta\right)}{\left(1-\eta^2\right) \left(1+2\left(1-\eta^2\right)g\left(N_S,N_B\right)\right)\left(1+\left(1-\eta^2\right)g\left(N_S,N_B\right)\right)^2} > 0
\end{equation}
because
\begin{align}
f_2\left(N_S,N_B,\eta\right) = \quad & 4N_S^2\left(1+2N_B\right)\eta^2\left(1-\eta^2\right)^2+8N_S\left(1-\eta^2\right)\left( N_B^2\left(1+\eta^2\right)^2+N_B\left(1+3\eta^2\right) +\eta^2 \right) \nonumber \\
 &+ 4N_B^2\left(1+\eta^2\right)\left(1-\eta^4\right)+4N_B\left(1+\eta^2\left(2-\eta^2\right)\right)+4\eta^2 > 0.
\end{align}
That is, the QFI is locally an increasing function of $\zeta$. The line of $\zeta=1$ is at the boundary of the parameter space. Therefore, the point $(1,1)$ is a maximum also with respect to $\zeta$.
\end{proof}

We strengthen Proposition~\ref{prop:maxQFI}  and show that the maximum at $(\zeta=1,r=1)$ is indeed the  global maximum in the $\eta\to0$ and $\eta\to1$ limits.
In the $\eta\to0$ case, the QFI is  monotone with respect to  $r$. This simplifies the optimization with respect to $\zeta$. Indeed, we have that
\begin{equation}
     \lim_{\eta\rightarrow 0} \left(\partial_r I_\eta^{\rm EA}\right)  =\frac{ 128g\left(N_S\zeta^2,N_B\right)N_B\left(N_B+ 1\right) \left(2N_S\zeta^2+1\right)^2r\left(1-r^4\right) }{\left[4r^2\left(1+2g\left(N_S\zeta^2,N_B\right)\right)^2-\left(1+r^2\right)^2\right]^2}\geq 0.
\end{equation}
This implies that, for $N_B>0$, $I_\eta$ is an increasing function of $r$, with $r=0$ and $r=1$ the only stationary points, where $r=0$ implies infinite squeezing. Because the gradient is strictly positive, $r=1$ is the optimal choice for any $\zeta$. We now study the gradient with respect to $\zeta$ and evaluate it along the line of $r=1$, as
\begin{equation}
   \lim_{\eta\rightarrow 0} \left(\partial_\zeta I_\eta^{\rm EA}\right)_{|{r=1}} = \frac{ 8N_S\zeta N_B\left(N_B+1\right)}{\left(1+2N_B\right)\left(1+g\left(N_S\zeta^2,N_B\right)\right)^2} \geq 0.
\end{equation}
The only stationary point is at $\zeta=0$, which is a minimum. Therefore, if  $N_B>0$, $\zeta=1$ is optimal. Furthermore, there are globally no other stationary points, so $(1,1)$ is the global maximum as $\eta=0$.

In the $\eta\rightarrow 1$ case, the asymptotic behaviour is
\begin{equation}\label{I2modeeta1}
     I_\eta^{\rm EA}  \sim \frac{2(N_S\zeta^2+N_B+2N_BN_S\zeta^2)}{1-\eta},\quad \mathrm{as}\ \eta\rightarrow 1,
\end{equation}
This expression is independent of $r$ and a growing function of $\zeta$. This implies the optimal strategy consists of allocating all photons to covariance. However, local squeezing, correlations, and any combination of the two perform equivalently. In fact, Eq.~\eqref{I2modeeta1} at $\zeta=1$ is identical to behaviour of the single-mode squeezed-vacuum, see Eq.~\eqref{asymp:sq}.


\subsection{QFI of the TMSV state}

The QFI of the TMSV can be written as 
\begin{equation}
    I_\eta^{\rm TMSV} =  
    \frac{ 4\left[N_S\left(N_S+1\right)\left(1-\eta^2\right)  + \eta^2\left(N_S+N_B+2N_SN_B\right)\right]}{\left(1-\eta^2\right)\left[1+\left(1-\eta^2\right)\left(N_S+N_B+2N_SN_B\right)\right]}.
\end{equation}
First, we notice that for $N_B=0$ the expression notably simplifies as $I_\eta^{\rm TMSV,(0)}=\frac{4N_S}{1-\eta^2}$,
which is clearly larger than any single-mode QFI as it saturates the bound in Lemma~\ref{bound}. Indeed, the TMSV state is an optimal probe for $N_B=0$ among the generic states (even non-Gaussian)~\cite{Nair2018}. However, the TMSV state does not perform asymptotically better than the optimal single-mode state for $N_B=0$. This can be seen by comparing directly with Eq.~\eqref{coheNSlarge}. 

For a generic $N_B$, we have that 
\begin{equation}\label{TMSV:NSlarge}
    I_\eta^{\rm TMSV}=\frac{4N_S}{(1-\eta^2)(1+2N_B)} + \mathcal{O}(1), \quad \mathrm{as}\ N_S\rightarrow \infty.
\end{equation}
In the large power regime, the optimal QFIs for single-mode and the TMSV perform virtually the same, as one can see by comparing Eq.~\eqref{sq:NSlarge} with Eq.~\eqref{TMSV:NSlarge}. The squeezed-vacuum state approaches the performance of the TMSV in the $\eta\to1$ limit, see Eq.~\eqref{I2modeeta1}. However, the TMSV state performs better on a larger region around $\eta=1$, as shown in Fig.~\ref{fig4}. For $\eta\to0$, the TMSV performs the same as a coherent state in the zero temperature case~\cite{Footnote2}. 
However, for increasing $N_B$, the quantum advantage approaches 2 for $N_S\lesssim1$, see Fig.~\ref{fig5}. Indeed, we have that 
\begin{equation}\label{TMSV:NSlow}
    I_\eta^{\rm TMSV}=I_{\rm sh}+\frac{4N_S}{1+(1-\eta^2)N_B}+\mathcal{O}(N_S^2), \quad \mathrm{as}\ N_S\to0.
\end{equation}
In the $1\gg N_S\gg N_B\eta^2$ and $N_B\gg1$ regime, we have an advantage of a factor of $2$ with respect to a optimized single-mode probe. This is a known result in the context of quantum illumination~\cite{Sanz2017, Nair2020, DiCandia2021}. 

\begin{figure}[t!]
    \centering
    \includegraphics[width=0.4\linewidth]{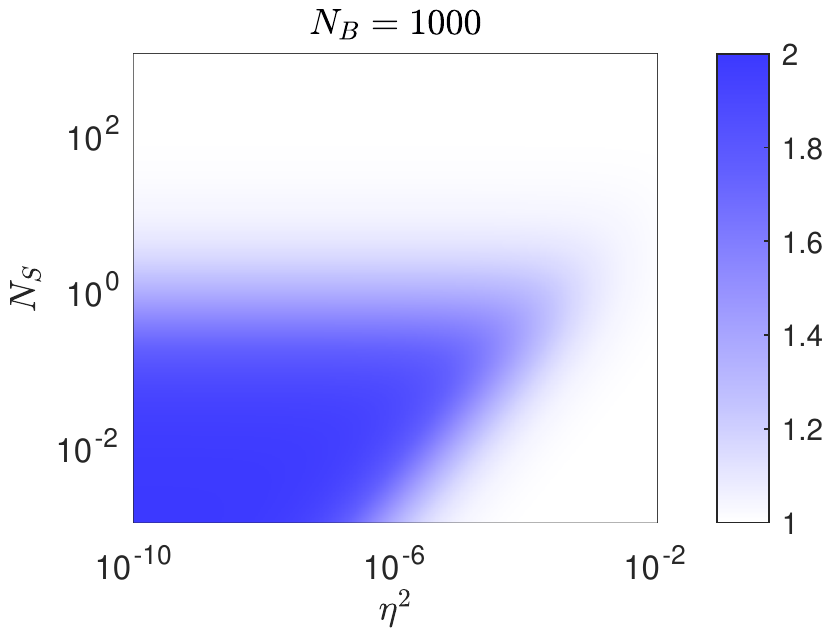}
    \caption{ Ratio between the QFIs of the TMSV and the coherent state, in the noisy ($N_B=10^3$) and lossy ($\eta\lesssim 10^{-2}$) regime.}
    \label{fig5}
\end{figure}

\section{Optimal total QFI}\label{sec:V}

Let us now discuss the case of optimizing the total QFI $\mathcal{I}_\eta=MI_\eta$ for fixed total power $\mathcal{N}_S=MN_S$. This analysis is relevant when we have a freedom of choosing how many copies of the states we will use. We shall notice that, in a continuous-variable experiment, the number $M$ can be increased by either repeating the experiment or by increasing the bandwidth. The latter, indeed, corresponds to performing several experiment in parallel.

Here, we have a clear distinction between the $N_B=0$ and the $N_B>0$ cases, due to the presence of the shadow-effect in the latter case. This is a power-independent term, that makes the total QFI optimized for $M=\infty$ if $N_B>0$. Indeed, if we have a constraint on the total power, then the larger the bandwidth the better is the achievable precision. This effect is similar to what happens in the quantum estimation of the amplifier gain, as analysed in Ref.~\cite{Nair2021}. In the amplifier case, this happens also at zero temperature, as amplification is an active operation for any temperature value.

In the following, we focus the discussion on the following aspects. We first solve the $N_B=0$ case: We show that either $M=1$ or $M=\infty$ is optimal in the idler-free case, while the choice of $M$ is irrelevant for the TMSV state. In the $N_B>0$ case, we consider an alternative model based on the environment normalization $N_B\rightarrow N_B/(1-\eta^2)$. This model has been widely used for studying remote quantum sensing scenario, such as quantum illumination and quantum reading. We show that, while without normalization a quantum advantage can be obtained only for $N_B\gg1$ and $\eta\ll1$, the normalization allows for an extension of the quantum advantage to any value of $\eta$. In this sense, we observe that the ultimate bound found by Nair and Gu~\cite{Nair2020}, can be reached for {\it any} $\eta$ by a TMSV state transmitter in the limit of infinite $M$.

In the following, similar notation as in the single probe QFI will be used. Indeed, we will refer the $N_B=0$ case with the suffix ``(0)''. Moreover, we will denote the normalized case with the suffix ``norm''. 

\subsection{The zero temperature case: $N_B=0$}
The following general bound will be useful for our discussion.
\begin{lemma}~{\bf \cite{Nair2018}}\label{bound1}
The total QFI of a generic multi-mode probe is bounded as $\mathcal{I}_\eta^{\rm (0)}\leq \frac{4\mathcal{N}_S}{1-\eta^2}$ for any $\eta\in(0,1)$ and total power $\mathcal{N}_S\geq0$.
\end{lemma}
For a coherent state probe, the number of probes $M$ is irrelevant for the performance in terms of total QFI, given that  $\mathcal{I}^{\rm (0)}_{\eta}(\xi=0)=4\mathcal{N}_S$. The situation changes when squeezing enters into the game. For instance, let us set $\xi=1$:
\begin{align}\label{sqM}
    \mathcal{I}_\eta^{\rm (0)}(\xi=1)=\frac{4\mathcal{N}_S[(1-\eta^2)^2+\eta^4]}{(1-\eta^2)[1+2\frac{\mathcal{N}_S}{M}\eta^2(1-\eta^2)]}.
\end{align}
We have that $\mathcal{I}_\eta^{\rm (0)}(\xi=1)>\mathcal{I}_\eta^{\rm (0)}(\xi=0)$ provided that $\frac{\mathcal{N}_S}{M}<\frac{2\eta^2-1}{2(1-\eta^2)^2}$. 
There are a couple of striking facts. First, if $\mathcal{N}_S\leq \bar N_S(\eta)$ (defined in Prop.~\ref{propNS}), then $M=\infty$ optimizes  the total QFI, and the squeezed-vacuum is an optimal probe. This is a direct consequence of Prop.~\ref{propNS}. Second, if $\eta>\frac{1}{\sqrt{2}}$, {\it for any total power} $\mathcal{N}_S$ we can choose a sufficiently large $M$ such that squeezed-vacuum does better than a coherent state. However, by using Eq.~\eqref{coheNSlarge}, we find that applying an infinitesimal squeezing to a largely displaced mode  virtually saturates the bound in Lemma~\ref{bound1}:
\begin{align}\label{totalQFI:NSinf}
    \mathcal{I}^{0}_\eta= \frac{4\mathcal{N}_S(1-\xi)}{1-\eta^2}+\mathcal{O}(1), \quad \mathrm{as}\ \xi\mathcal{N}_S\rightarrow \infty.
\end{align}
We now show the result for the optimal bandwidth given a certain amount of power at disposal.
\begin{prop}
The total QFI $\mathcal{I}_\eta^{\rm (0)}$ is optimized either for $M=1$ or $M=\infty$.
\end{prop}
\begin{proof}
Let us denote $N_S=\mathcal{N}_S/M$, and extend, for simplicity, the optimization problem to the continuum. Indeed, we consider $N_S\in [0,\mathcal{N}_S]$. We are interested in the $N_S$ value that solves the problem
\begin{align}\label{opt1}
    \max_{N_S\in[0,\mathcal{N}_S],\xi\in[0,1]} \frac{\mathcal{I}_\eta^{\rm (0)}}{4\mathcal{N}_S}=\max_{\xi\in[0,1]}\left\{\max_{N_S\in[0,\mathcal{N}_S]}\frac{\mathcal{I}_\eta^{\rm (0)}}{4\mathcal{N}_S}\right\}.
\end{align}
After a change of variable $\xi N_S= x$, we have that
\begin{align}\label{optM}
    \max_{\xi\in[0,1]}\left\{\max_{N_S\in[0,\mathcal{N}_S]}\frac{\mathcal{I}_\eta^{\rm (0)}}{4\mathcal{N}_S}\right\}&=\max_{0\leq x \leq N_S^{\rm opt}(x)}\left\{\max_{N_S\in[0,\mathcal{N}_S]}\left[ \frac{1-\frac{x}{N_S}}{1-2\eta^2\left(\sqrt{x(1+x)}-x\right)}+\frac{\frac{x}{N_S}\left[(1-\eta^2)^2+\eta^4\right]}{(1-\eta^2)(1+2x\eta^2(1-\eta^2))}\right]\right\} \\
    \quad&\equiv \max_{0\leq x \leq N_S^{\rm opt}(x)}\left\{\max_{N_S\in[0,\mathcal{N}_S]} h_\eta(x,N_S)\right\},
\end{align}
where $N_S^{\rm opt}(x)$ is the argmax of the optimization with respect to $N_S$.
The function $h_\eta(x,N_S)$ is linear in $N_S^{-1}$, meaning that the maximum is in one of the extreme point, \textit{i.e.}, $N_S^{\rm opt}(x)$ is either $0$ or $\mathcal{N}_S$~\cite{Footnote3}. 
This means that either $M=1$ or $M=\infty$ is the optimal choice.
\end{proof}
\begin{figure}
    \centering
    \includegraphics[width=0.9\linewidth]{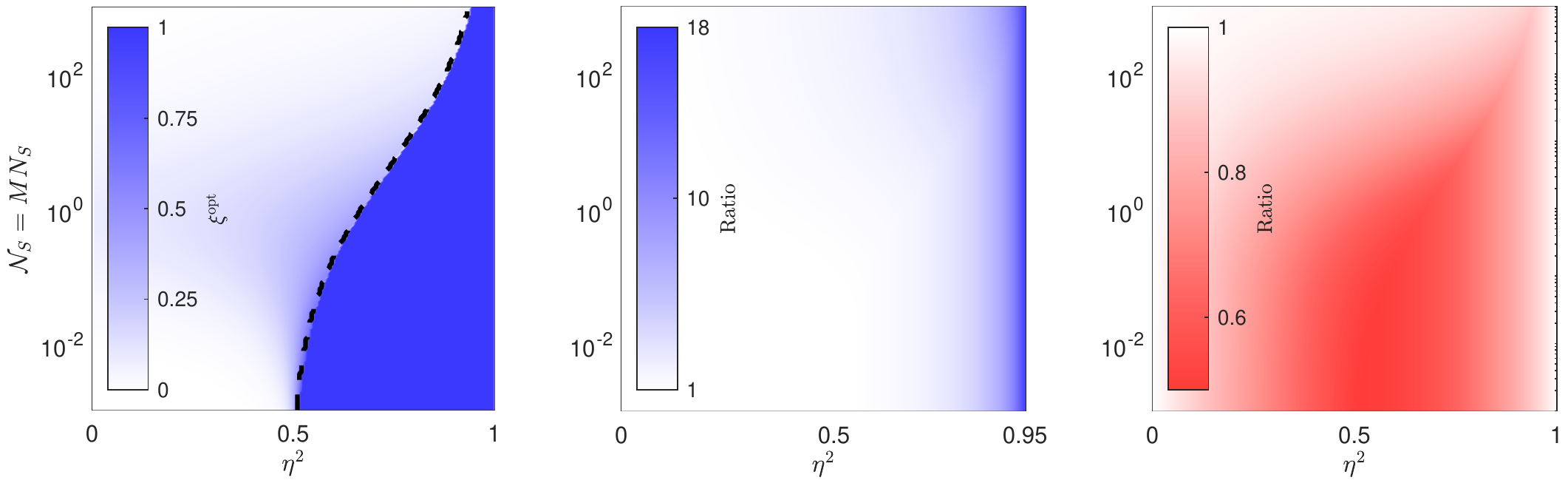}
    \caption{(\textbf{Left}) Optimal $\xi$ for the single-mode QFI, jointly optimized over the bandwidth $M$ for a total power $\mathcal{N}_S=MN_S$. The dashed line indicates the switch from $M^{\rm opt}=1$ (left) and $M^{\rm opt}=\infty$ (right). Here, $\xi^{\rm opt}=1$ on a larger region with respect to Fig.~1a. We have two clear regions corresponding to $\{M^{\rm opt}=\infty,\xi^{\rm opt}=1\}$ and  $\{M^{\rm opt}=1,\xi^{\rm opt}<1\}$. (\textbf{Middle}) Ratio of the QFI for the optimized idler-free state and the coherent state. (\textbf{Right}) Ratio of the QFI for the optimized  idler-free state and the TMSV state. }
    \label{fig3}
\end{figure}

Notice that in the limit of large total power, the total QFI is virtually optimized for any $M$. This is clear from Eq.~\eqref{totalQFI:NSinf}, which does not depend on $M$. The next question is whether squeezed-vacuum states perform better than any state for fixed total power. This turns out to depend on the total available energy, as shown the following Proposition.
\begin{prop}
There exists $\bar K(\eta)$ such that $M=1$ optimizes $\mathcal{I}_\eta^{\rm (0)}$ for any $\mathcal{N}_S> \bar K(\eta)$. We have that $\bar K(\eta)=0$ for $0<\eta\leq\frac{1}{\sqrt{2}}$ and $\bar K(\eta)\geq \bar N_S(\eta)$ for $\eta>\frac{1}{\sqrt{2}}$.
\end{prop}
\begin{proof}
Let us consider $0<\eta\leq\frac{1}{\sqrt{2}}$. For $M=\infty$, the total QFI $\mathcal{I}_\eta^{\rm (0)}=4\mathcal{N}_S\left(1-\xi+\xi\frac{(1-\eta^2)^2+\eta^4}{1-\eta^2}\right)$ is optimized for $\xi=0$, \textit{i.e.}, for a coherent state probe. Notice that the performance of a coherent state probe is the same for any $M$, \textit{i.e.}, $\mathcal{I}_\eta^{\rm (0)}(M=\infty,\xi=0)=\mathcal{I}_\eta^{\rm (0)}(M,\xi=0)$ for any {\it finite} $M$. However, due to Proposition~\ref{propCOH}, for any finite $M$ there is a squeezed coherent state that performs better than a coherent state probe, which is an absurd. It follows that $M=\infty$ cannot optimize the total QFI. In this case, $M=1$ is optimal for any $\mathcal{N}_S>0$.

Let us now consider $\eta>\frac{1}{\sqrt{2}}$. Let us extend the optimization domain to $N_S\in[0,\infty]$. The quantity $\frac{\mathcal{I}_\eta^{\rm (0)}}{4\mathcal{N}_S}$ is maximal for $N_S=\infty$, as for this value the bound in Lemma~\ref{bound1} is saturated. This means that, in Eq.~\eqref{optM}, there exists $\bar K(\eta)$ such that $h_\eta(x,N_S)>h_\eta(x,0)$ for any $N_S\geq \bar K(\eta)$. It follows that if $\mathcal{N}_S>\bar K(\eta)$, then $M=1$ is optimal. 
In addition, we have that $\bar K(\eta)\geq N_S(\eta)$. In fact, if $\mathcal{N}_S\leq \bar N_S(\eta)$, then the optimal choice is $M=\infty$, as shown below Eq.~\eqref{sqM}. 
\end{proof}
In Fig.~\ref{fig3} we numerically show that $\tilde K(\eta)$ is {\it strictly} larger than $\tilde N_S(\eta)$. This is because when jointly optimizing the total QFI with respect to $M$ and $\xi$, the squeezed-vacuum state results to be the optimal choice on a larger range of parameter values. In this case we numerically see that $\xi^{\rm opt}=1$ if and only if $\mathcal{N}_S\leq\tilde K(\eta)$, and the optimal value is achieved in the limits $N_S\to0$ and $M\to\infty$, with the constraint $MN_S=\mathcal{N}_S$.

Regarding the TMSV case, we have that the total QFI $\mathcal{I}^{\rm TMSV}_\eta=\frac{4\mathcal{N}_S}{1-\eta^2}$, is independent on $M$. No advantage with respect an optimized single-mode transmitter can be observed in the $\mathcal{N}_S\gg1$ regime, as $\mathcal{I}^{\rm TMSV}_\eta$ approaches the optimal total QFI achieved in the idler-free case, see Eq.~\eqref{totalQFI:NSinf}. However, one shall keep in mind that reaching the performance of Eq.~\eqref{totalQFI:NSinf} needs squeezing, albeit an infinitesimal amount. Indeed, the TMSV still show an advantage with respect to a coherent state transmitter for any $\eta\not=0$. In addition, due to Lemma~\ref{bound1}, the TMSV state is indeed an optimal probe for any value of $\eta$. 
In Fig.~\ref{fig3}, it is shown a factor of $2$ advantage is reached for a large range of values of $\eta$ if $\mathcal{N}_S\lesssim 1$, and it decreases with increasing $\mathcal{N}_S$. Notice that for $\mathcal{N}_S\simeq1$ and $\eta\simeq\frac{1}{\sqrt{2}}$ we have that $\mathcal{I}_{\eta}^{\rm TMSV,(0)}\simeq 10$, which is enough to realize a sensitivity up to $\Delta \hat \eta^2\lesssim 0.1$. To achieve larger sensitivity values, the optimal displaced squeezed state shall be a better choice for an experimentalist, as it realizes similar performances as the TMSV probe.


\subsection{The finite temperature case: $N_B>0$}

\subsubsection{Quantum advantage with the shadow-effect}

As previously discussed, $M=\infty$ is the optimal choice for any value of $\mathcal{N}_S$, due to the presence of the shadow-effect. Let us discuss a limit where the shadow-effect is not present, and where the TMSV is expected to show a relevant advantage with respect to the single-mode case. 
In the finite $N_B$ case, we expect to have an advantage of the TMSV state over the idler-free strategy for low ({\it albeit} finite) values of $\mathcal{N}_S$, similarly as shown in the $N_B=0$ case.
Let us focus on the $N_B\gg(1-\eta^2)^{-1}$ regime. The presence of the shadow-effect makes the quantum advantage disappears for finite values of $\eta$. Therefore, we consider $N_B\eta^2 \ll \mathcal{N}_S/M$. In this regime, we have that $\mathcal{I}^{\rm TMSV}_\eta\simeq \frac{4\mathcal{N}_S(M+\mathcal{N}_S)}{N_B(M+2\mathcal{N}_S)}$, while for the coherent state we get $\mathcal{I}_\eta^{\rm coh}= \frac{2\mathcal{N}_S}{N_B(1-\eta^2)}$. It is  clear  that $\mathcal{I}_\eta^{\rm TMSV}$ is optimized for $M\gg \mathcal{N}_S$, which also implies that $\eta^2 N_B$ must be much smaller than 1. This agrees with the analysis done after Eq.~\eqref{TMSV:NSlow}. In this regime, the TMSV state shows a quantum advantage of 2 for arbitrarily large $\mathcal{N}_S$. In Fig.~\ref{fig5}, the $M=1$ is drawn. It is visible that the quantum advantage is present for $\eta^2 N_B\ll1$, and it disappears already for $\eta^2 N_B\sim 1$.

\subsubsection{Erasing the shadow-effect: $N_B\rightarrow N_B/(1-\eta^2)$}
This normalization has been used for discussing remote sensing protocols such as quantum illumination and quantum reading. It erases the shadow-effect, and, with that, any benefit derived by its presence. In this case, the following result has been proved by Nair and Gu.
\begin{lemma}~{\bf \cite{Nair2020}}\label{bound2}
The total QFI of a generic multi-mode probe in the normalized environment case is bounded as $\mathcal{I}^{\rm norm}_\eta\leq \frac{4\mathcal{N}_S}{N_B+1-\eta^2}$ for any $\eta$ and total power $\mathcal{N}_S\geq0$.
\end{lemma}
In Appendix~\hyperref[app:C]{C}, we have computed the QFI with the noise normalization, for both the single-mode and the TMSV state case. In the idler-free case, we have that 
\begin{align}\label{normQFIlargeNS}
    \mathcal{I}_\eta^{\rm norm, IF} = \frac{4\mathcal{N}_S(1-\xi)}{2N_B+1-\eta^2} +\mathcal{O}(1),  \quad \mathrm{as}\ \xi\mathcal{N}_S\rightarrow \infty
\end{align}
for any $M$. For a coherent state input, {\it i.e.}, for $\xi=0$, we get that $\mathcal{I}_\eta^{\rm norm, coh}=\frac{4\mathcal{N}_S}{1+2N_B}$, meaning that an infinitesimal amount of squeezing allows us  to reduce the QFI as in Eq.~\eqref{normQFIlargeNS}. Comparing this result with Eq.~\eqref{sq:NSlarge}, we see that the $(1-\eta^2)^{-1}$ divergence disappears. Indeed, in the $N_B\gg1$ regime, the un-squeezed coherent state is virtually the optimal probe for any value of $\eta$.

The QFI of the TMSV state can be written as 
\begin{equation}
 {\mathcal{I}}_\eta^{\rm norm, TMSV} = \frac{4\mathcal{N}_S\left[ N_B+1 + \mathcal{N}_S\left(N_B+1-\eta^2\right)M^{-1}   \right]}{\left(N_B+1-\eta^2\right)\left[ N_B+1 + \mathcal{N}_S\left(2N_B+1-\eta^2\right)M^{-1}   \right]}.
\end{equation}
In the infinite bandwidth limit we get
\begin{equation}\label{limitTMSVnorm}
{\mathcal{I}}_\eta^{\rm norm, TMSV}  = \frac{4\mathcal{N}_S}{N_B+1-\eta^2} +\mathcal{O}(M^{-1}),  \quad \mathrm{as}\ M\rightarrow \infty.
\end{equation}
Equation~\eqref{limitTMSVnorm} saturates the ultimate bound in Lemma~\ref{bound2} for any value of $\eta$. In the normalized environment case, the infinite bandwidth TMSV state is indeed the optimal probe {\it for any value of $\eta$}~\cite{Footnote4}. 
The quantum advantage is limited to a factor of $2$ in the QFI, and is obtained in the limit of large $N_B$.

We notice that there is a clear qualitative distinction between in the normalized and the unnormalized models. In the unnormalized model, the shadow-effect washed out the quantum advantage for low-enough $N_S$. A quantum advantage is reached by the TMSV state only when the bandwidth of the classical probe is limited, and for large enough power per mode. Instead, in the normalized model, the TMSV state shows a quantum advantage for any parameter value, unless $N_B=0$.

\section{Quantum hypothesis testing}\label{sec:VI}
Quantum hypothesis testing is the discrete version of quantum parameter estimation. It consists in the discrimination between two values of a system parameter, by sending a quantum state as probe. Given a $\eta$-dependent channel $\mathcal{E}_\eta$, discriminating between the values $\eta=\eta_+$ and $\eta=\eta_-$ ($\eta_+>\eta_-$) using $M$ copies of the state $\rho$ as a probe results in the average error probability
\begin{equation}\label{proberr}
     P_{\rm err}=\frac{1-\frac{1}{2}\|\rho_{\eta_+}^{\otimes M}-\rho_{\eta_-}^{\otimes M}\|_1}{2},
\end{equation}
where $\rho_\eta=\mathcal{E}_\eta[\rho]$, $\|\cdot\|_1$ is the trace norm. Here, we have assumed equal {\it a-priori} probabilities for the two hypotheses, but the discussion can be trivially generalized to the asymmetric setting. Generally, the quantity in Eq.~\eqref{proberr} is challenging to compute. In addition, saturating the equality in Eq.~\eqref{proberr} requires one to collectively measure the $M$ output copies of the channel, which in most cases is not implementable with current technology. In the following, we discuss a simple bound based on the QFI.

We first recall that the QFI can be generally written as 
\begin{align}\label{QFI_fidelity}
    I_\eta = \lim_{d\eta \rightarrow0}\frac{8}{d\eta^2}\left[1-\sqrt{F(\rho_\eta,\rho_{\eta-d\eta})}\right],
\end{align}
where $F(\rho,\sigma)=[\text{Tr}\,(\sqrt{\rho\sqrt{\sigma}\rho})]^2$ is the fidelity between the states $\rho$ and $\sigma$. We can now use this relation to bound the optimal discrimination error probability as~\cite{Pirandola2017}
\begin{align}\label{QFI_hyp_bound}
    P_{\rm err}\leq \frac{1}{2}\sqrt{\left[F(\rho_{\eta_+},\rho_{\eta_-})\right]^{M}}\simeq \frac{1}{2}\mathrm{e}^{-Md\eta^2 I_\eta/8},
\end{align}
where we have defined $d\eta=\eta_+-\eta_-$, and the approximation holds for $d\eta^2I_\eta \ll1$. The bound in Eq.~\eqref{QFI_hyp_bound} is achievable by measuring the $M$ copies of the output state separately, and then applying a threshold discrimination strategy~\cite{Calsamiglia2008, Sanz2017}. We can optimally estimate the parameter $\eta$, obtaining a value $\eta_{\rm est}$. We then decide towards the hypothesis $\eta=\eta_+$ if $\eta_{est}>kd\eta$ with $0<k<1$, or $\eta=\eta_-$ otherwise. If $\eta_+$ and $\eta_-$ are sufficiently close, then the optimal choice is $k=1/2$. If $M$ is large enough, the error probability can be approximated as $P_{\rm err}\simeq 1-\text{erf}~(\sqrt{d\eta^2 I_\eta M/8})\simeq \frac{1}{2}\mathrm{e}^{-Md\eta^2 I_\eta/8}$ for large enough $Md\eta^2I_\eta$. This strategy saturates the bound in Eq.~\eqref{QFI_hyp_bound}.

An important observation is about the number of copies needed to achieve an exponential decay of the fidelity and the error probability in the input power. The fidelity between two $n$-mode Gaussian quantum states has the following structure:
\begin{align}\label{fidelity:structure}
    F(\rho,\sigma)\sim \frac{1}{{\rm poly}(n)}\exp\left\{{-{\bf \delta}^\top[{\bf \Sigma}_\rho+{\bf \Sigma}_\sigma]^{-1}{\bf \delta}}\right\},
\end{align}
where ${\bf \delta} = {\bf d}_\rho-{\bf d}_\sigma$ is the displacement difference between the two states, and ${\rm poly}(n)$ is a polynomial of degree $n$ dependent solely on the covariance matrices~\cite{Banchi2015}. For finite $n$ ({\it e.g.}, $n=2$), in order to have an exponential decay of the error probability with respect to $\mathcal{N}_S$, we need at least one of the following two properties to be fulfilled: (i) A non-zero displacement; (ii) An infinite number of probes ($M=\infty$). It follows that squeezed-vacuum and TMSV states can have an exponential decay of the error probability only in the infinite bandwidth case. As we have shown in the previous sections, the total QFI of these states is actually maximized for $M=\infty$. For a coherent state input, the error probability performance does not depend on the bandwidth choice. For an optimized displaced squeezed state, the QFI of the unnormalized model shows a divergence for $\eta=1$, as shown in Eqs.~\eqref{coheNSlarge}-\eqref{sq:NSlarge}. 
This divergence can be seen at the error exponent for any choice of $M$ (including $M=1$), as one can readily check using Eq.~\eqref{fidelity:structure}. Indeed, for these two states,
Eq.~\eqref{QFI_hyp_bound} holds also for $M=1$ and $d\eta^2I_\eta\gtrsim1$.
In the following, we consider the lossy channel introduced in Eq.~\eqref{Lindblad_lossy}, that can be rewritten as $\mathcal{E}_\eta=\mathrm{e}^{-2\ln(\eta)\mathcal{L}}$, where $\mathcal{L}[\rho]=(1+N_B)\mathcal{D}(a_S)[\rho]+N_B\mathcal{D}(a_S^\dag)[\rho]$. We now discuss the exemplary cases of quantum illumination and quantum reading.

\subsection{Quantum illumination}

In quantum illumination (QI), a faster decay rate in the probability of error can be achieved with an entangled probe. This makes QI an important illustration of a quantum advantage that ``survives''  an entanglement breaking channel, which is the case for $\eta^2<\frac{N_B}{1+N_B}$. In particular, a lot of interest has been raised for $N_B\gg1$ and $\eta^2\ll1$, where the TMSV state shows a relevant quantum advantage in the error probability exponent~\cite{Tan2008}. Indeed, this may have applications in radar-like remote sensing in the microwave regime~\cite{Barzanjeh2015, LasHeras2017}, where $N_B$ is of the order of thousands of photons in a room temperature environment.
As shown in the previous section, in this regime the TMSV can realize at most a factor of $2$ advantage in QFI over a coherent state, observed by choosing a low power-per-mode regime for the TMSV state ($\mathcal{N}_S/M=N_S\ll 1$), see Fig.~\ref{fig5}~\cite{Sanz2017,DiCandia2021}. This advantage is observed by $\frac{I_\eta^{\rm TMSV}}{I_\eta^{\rm coh}} = 1+\left(2N_S+1+\frac{N_S+1}{N_B}\right)^{-1}+\mathcal{O}(\eta^2)$. 
%

QI is usually studied in a modified setting, with a constant background for all transmissions, \textit{i.e.}, $N_B \rightarrow \frac{N_B}{1-\eta^2}$. This change is done {\it ad hoc} to eliminate the shadow-effect, as shown in the previous section. In a radar scenario, the shadow-effect can be interpreted as an artifact of the considered model~\cite{Tan2008,TanPHD}, which for finite $\eta$ can be relevant. There are regimes where the presence of the shadow-effect is not relevant in QI. Indeed, the normalized and unnormalized models perform the same for $\eta^2N_B\gg N_S\gg1$. However, for $N_S\lesssim \eta^2N_B$ the shadow-effect starts to be relevant for the QFI value. This has consequences also on the optimal receiver. In fact, the optimal TMSV receiver for the normalized model is either a phase conjugate (PC) receiver or an optical parametric amplifier (OPA) receiver~\cite{Guha2009,Sanz2017}. Instead, for the unnormalized model in the $N_S\lesssim \eta^2N_B$ regime, a double homodyne receiver performs better than both the OPA and PC receivers, as shown in Ref.~\cite{Jo2021}. This consideration holds for both the hypothesis testing and the parameter estimation problems. Understanding the right way to model a QI scenario is thus of crucial importance for experiments. Indeed, this shall be done by analyzing a realistic quantum model of wave propagation theory.

Lastly, we notice that the quantum advantage achieved by a TMSV probe is restricted to $\mathcal{N}_S/M=N_S\lesssim1$. Indeed, consider a TMSV where a quantum-limited, large amplification is applied to the signal. In other words, consider the signal mode $a'_S=\sqrt{G}~a_S(t)+{\sqrt{G-1}}~v$, where $G\gg1$ and $v$ is a vacuum mode. Since the advantage is limited to $N_S\lesssim1$, see Fig.~\ref{fig5}, this amplification process adds enough noise to destroy the quantum advantage of $2$ in the SNR. This agrees with the analysis done in Refs.~\cite{Shapiro, Jonsson2020}.

\subsection{Quantum reading}

Quantum reading consists in embedding a bit of information in the reflectivity parameter $\eta$ of a cell~\cite{Pirandola2011}. Since this is thought to be implemented in a controlled environment, the hypothesis testing to retrieve the information is between two different values of $\eta$ close enough to $1$. In the optical case, {\it i.e.}, for $N_B=0$, our results recognize regimes where the squeezed-vacuum state is optimal in discriminating between two values 
of $\eta$ close to each other. Indeed, of particular interest is the analysis done in Fig.~\ref{fig3}, see Section~\hyperref[sec:V]{V} for a discussion. As already mentioned, for large enough values of $\mathcal{N}_S$ an optimized displaced squeezed state shall be the best choice for an experimentalist to get an advantage with respect to a coherent state probe, as both the transmitter and the receiver are less experimentally challenging to implement. The situation changes in the bright environment case, {\it i.e.}, $N_B\gg1$. Here, understanding what model actually describes the experiment is of crucial importance, as the normalized and the unnormalized models give radically different results. The differences are even more evident than in the QI case. In the unnormalized model, we have that the optimal idler-free and entanglement-assisted states show a relevant quantum advantage only when $\eta$ is enough close to $1$, see Fig.~\ref{fig4}. However, this quantum advantage is potentially unbounded, since it relies on the presence of the $(1-\eta^2)$-divergence. Instead, by normalizing the environment with $N_B\rightarrow N_B/(1-\eta^2)$, the TMSV state shows a quantum advantage for any value of $\eta$ with respect to both the coherent state and the optimal idler-free probes, see Eqs.~\eqref{normQFIlargeNS}-\eqref{limitTMSVnorm}. However, this advantage is limited to a factor of 2 in the QFI, achieved for large enough $N_B$.

\section{Conclusion}

In this article, we have characterized the metrological power of energy-constrained Gaussian state probes in the task of estimating the loss parameter of a thermal channel. We have showed that, with access to an entangled idler, the two-mode squeezed-vacuum state is the optimal probe in all regimes. Conversely, in the idler-free scenario, we have showed that the optimal state is generally a non-trivial trade-off between displacement and local squeezing. We have provided analytical results aimed to understand the behaviour of the optimal state in the finite parameter regime. We have considered the problem of optimizing the {\it total} quantum Fisher information, with a constraint on the total input energy. In this context, we have analysed the role of the shadow-effect in getting a quantum advantage, defined by using either single-mode or two-mode squeezing for the state preparation. In addition, we have recognized the main differences between considering the bare lossy channel, and a corresponding normalized channel widely used in remote sensing scenarios. We have shown that a TMSV probe is the optimal probe for both of these channels. However, its advantage with respect to the idler-free case is present for any parameter value only in  the normalized model. We have related these results to topical discrimination protocols, such as quantum illumination and quantum reading. Our results aim to elucidate important aspects of the sensing performance in Gaussian-preserving bosonic channels with both analytical and numerical insights.

\section*{Acknowledgements}
RDC acknowledges support from the
Marie Sk{\l}odowska Curie fellowship number 891517 (MSC-IF
Green-MIQUEC). RJ acknowledges support from the Knut  and  Alice  Wallenberg  (KAW) foundation  for  funding through the Wallenberg Centre for Quantum Technology (WACQT).

\section*{Appendix~A: Idler-free case}\label{app:A}

In this appendix, we discuss several technical details used for deriving the idler-free results in Section~\hyperref[sec:III]{III}. 

\subsection*{A1: Single-mode QFI}\label{app:A1}

Here, we find the QFI used Eq.~\eqref{QFIsingle} for a generic probe in a displaced squeezed-vacuum state parametrized as in Eq.~\eqref{single}. As discussed in the main text, we prepare the probe state with displacement along $\theta = n\pi$ such that $\mathbf{d}=\left(\sqrt{2N_{\rm coh}},0\right)^T$ and  the covariance matrix is ${\bf \Sigma}_S={\rm diag}\left(\frac{r}{2},\frac{1}{2r}\right)$ with  $r=1+2N_{\rm sq}-2\sqrt{N_{\rm sq}(N_{\rm sq}+1)}$. The QFI can be written as sum of three terms:
\begin{align}
    I_\eta^{\rm IF}\left(N_{\rm coh},N_{\rm sq}\right)= I_{\eta,1}^{\rm IF}\left(N_{\rm coh},N_{\rm sq}\right) + I_{\eta,2}^{\rm IF}\left(N_{\rm sq}\right) + I_{\eta,3}^{\rm IF}\left(N_{\rm sq}\right).
\end{align}
These terms are defined as
\begin{align}
    I_{\eta,1}^{\rm IF}\left(N_{\rm coh},N_{\rm sq}\right) &= \frac{4N_{\rm coh}}{ \eta^2r +\left(1-\eta^2\right)\left(2N_B+1\right)},\\
   I_{\eta,2}^{\rm IF}\left(N_{\rm sq}\right) &= \frac{4N_{\rm sq}\eta^2\left(2N_B+1\right)}{A}\left\{ \frac{\left(N_{\rm sq}+1\right)\left(2N_B+1\right)}{2A+1}-1 \right\},\\
    I_{\eta,3}^{\rm IF}\left(N_{\rm sq}\right) &= \frac{4N_B^2\eta^2}{A}, 
\end{align} 
 where  $A = \left(1-\eta^2\right)\left[N_B\left(N_B+1\right)+N_{\rm sq}\eta^2\left(2N_B+1\right)-N_B^2\eta^2\right]$. 

The coherent state performance is due to two terms, {\it i.e.} $I_\eta^{\rm coh} = I_{\eta,1}\left(N_{\rm coh},0\right) +  I_{\eta,3}\left(0\right)$. The
term $I_\eta^{\rm shad} = I_{\eta,3}\left(0\right)$ gives the shadow-effect.
For a squeezed-state probe, we have that $I^{\rm sq}_\eta= I_{\eta,2}\left(N_{\rm sq}\right) + I_{\eta,3}\left(N_{\rm sq}\right)$.
In the large power regime, this can be expressed as
\begin{align}
    I_{\eta}^{\rm sq}= \frac{2(1-\eta^2)^2+2\eta^4}{\eta^2(1-\eta^2)^2}+\mathcal{O}(N_S^{-1}), \quad N_S\to \infty,
\end{align}
which saturates to a $\eta$-dependent value.

\subsection*{A2: Properties of $I_\eta^{\rm IF, (0)}$}\label{app:A2}
\subsubsection{Concavity of $I_\eta^{\rm IF, (0)}$}
We have that
\begin{align}
   \frac{1}{2N_S^2\eta^2}& \partial^2_\xi I_\eta^{\rm IF, (0)} = -\frac{8(1-2\eta^2+2\eta^4)}{(1+2N_S \xi\eta^2(1-\eta^2))^3} +\frac{l_1(N_S,\eta,\xi)}{\sqrt{N_S \xi(1+N_S\xi)}[1-2\eta^2(\sqrt{N_S\xi(1+N_S\xi)}-N_S\xi)]^2},
\end{align}
with
\begin{align}
    l_1(N_S,\eta,\xi)&=\frac{2N_S(1-\xi)\eta^2[1-2(\sqrt{N_S\xi(1+N_S\xi)}-N_S\xi)]^2}{\sqrt{N_S\xi(1+N_S\xi)}[1-2\eta^2(\sqrt{N_S\xi(1+N_S\xi)}-N_S\xi)]}-\frac{1-\xi}{\xi(1+N_S\xi)}-4[1-2(\sqrt{N_S\xi(1+N_S\xi)}-N_S\xi)].
\end{align}
We notice that the maximum value of $l_1$ is reached for $\eta=1$ and $N_S\xi\to \infty$, for which $l_1\sim -(1-\xi)/(2N_S\xi^2)$. Therefore $l_1<0$, and  $\partial_\xi^2 I_\eta^{\rm IF, (0)}$ is negative for any parameter values unless $\eta=0$.

\subsubsection{First derivative of $I_\eta^{\rm IF, (0)}$ with respect to $\xi$}
Let us compute $f(\eta,N_S,\xi)=\frac{1}{4N_S}\partial_\xi I_\eta^{\rm IF, (0)}$:
\begin{align}
    f(\eta,N_S,\xi)&=\frac{2N_S(1-\xi)\eta^2\left(\frac{1+2N_S \xi}{2\sqrt{N_S\xi(1+N_S\xi)}}-1\right)}{\left[1-2\eta^2\left(\sqrt{N_S \xi (1+N_S\xi)}-N_S\xi\right)\right]^2}-\frac{2N_S \xi \eta^2 (1-2\eta^2+2\eta^4)}{\left[1+2N_S\xi\eta^2(1-\eta^2)\right]^2} \nonumber \\
    \quad&+\frac{1-2\eta^2+2\eta^4}{(1-\eta^2)(1+2N_S\xi \eta^2(1-\eta^2))}-\frac{1}{1-2\eta^2\left(\sqrt{N_S\xi(1+N_S\xi)}-N_S\xi\right)}.
\end{align}
Let us now study the derivative in $\xi=1$, {\it i.e.}, $f_1(\eta, N_S)=\frac{1}{4N_S}\left(\partial_\xi I_\eta^{\rm IF, (0)}\right)_{|\xi=1}$ it and its derivative with respect to $N_S$:
\begin{align}\label{derNS}
    f_1(\eta,N_S)&= \frac{(1-\eta^2)^2+\eta^4}{(1-\eta^2)[1+2N_S\eta^2(1-\eta^2)]^2}-\frac{1}{1-2\eta^2(\sqrt{N_S(1+N_S)}-N_S)},\\
    \frac{1}{\eta^2}\partial_{N_S} f_1(\eta,N_S) &=-\frac{2\left(\frac{1+2N_S}{2\sqrt{N_S(1+N_S)}}-1\right)}{\left[1+2\eta^2\left(N_S-\sqrt{N_S(1+N_S)}\right)\right]^2}-\frac{4[(1-\eta^2)^2+\eta^4]}{\left[1+2N_S\eta^2(1-\eta^2)\right]^3}<0.
\end{align}
This means that $f_1$ is always decreasing in $N_S$. Notice that $f_1\to -\frac{1}{1-\eta^2}$ for $N_S\to\infty$, and that $f_1(\eta,N_S=0)=\frac{2\eta^2-1}{1-\eta^2}$ is positive for some $\eta>1/\sqrt{2}$.

The function $f$ is singular in $\xi=0$. Its expansion is 
\begin{align}
    f(\eta, N_S,\xi)=\frac{\sqrt{N_S}\eta^2}{\sqrt{\xi}}+\mathcal{O}(1), \quad \xi\to0.
\end{align}

\subsection*{A3: Perturbative analysis of $\bar \eta(N_S)$}\label{app:A3}

Let us find the asymptotic behaviour of $\bar \eta$ for small and large $N_S$.
In the $N_S\gg1$ regime, the expansion of $f$ shows a zero for $\bar \eta\simeq 1$. In order to get further asymptotic terms, we set $\bar \eta\sim 1-\frac{1}{c_1N_S}$ for some $c_1>0$. We have that
\begin{align}\label{eqc1}
    f(\bar \eta,N_S)\sim\frac{c_1(-128-64c_1+c_1^3)}{2(4+c_1)^2(8+c_1)}N_S.
\end{align}
By solving Eq.~\eqref{eqc1} to zero, we find one positive root $c_1\simeq8.86$. Let us now focus on the $N_S\ll1$ regime. Here, from the zeroth expansion of $f$ around $N_S=0$, we obtain a root for $\bar \eta \simeq \frac{1}{\sqrt{2}}$. By setting $\bar\eta\sim \frac{1}{\sqrt{2}}+c_2\sqrt{N_S}$ for some $c_2\in\mathbb{R}$, we obtain
\begin{align}\label{eqc2}
    f(\bar \eta,N_S)\sim2(2\sqrt{2}c_2-1)\sqrt{N_S}.
\end{align}
Solving Eq.~\eqref{eqc2} to zero, we get $c_2=\frac{1}{2\sqrt{2}}\simeq0.35$. Higher order expansions can be obtained by iterating this procedure. 

\subsection*{A4: Asymptotic expansion for $\xi^{\rm opt}$}\label{app:A4}



The Taylor expansion for large $\xi N_S$ is
\begin{align}
    I_\eta^{\rm IF}=\frac{4N_S(1-\xi)}{(1-\eta^2)(1+2N_B)}\left[1-\frac{\eta^2}{4N_S\xi(1-\eta^2)(1+2N_B)}\right]+\frac{2(1-2\eta^2+2\eta^4)}{\eta^2(1-\eta^2)^2}+{ \scriptstyle \mathcal{O}}(1),
\end{align}
which holds for $\xi N_S\gg\eta^2/[(1-\eta^2)(1+2N_B)]$ and $N_S\gg N_B$. By setting the derivative with respect to $\xi$ to zero and solving for $\xi$, we obtain $\xi^{\rm opt}\sim \eta/[4N_S(1-\eta^2)(1+2N_B)]^{1/2}$. This includes the $N_B=0$ case discussed in Eq.~\eqref{coheNSlarge}.

\subsection*{A5: Abrupt change of $\xi^{\rm opt}$ for $N_S\ll 1$}\label{app:A5}
The expansion of the QFI $I_\eta$ in the limit of small $N_S$ is
\begin{align}\label{exp_NSsmall}
    I_\eta^{\rm IF}=I_{\rm shad}+4N_S\left\{\frac{1-\xi}{1+2N_B(1-\eta^2)}+\xi g_2(\eta,N_B)\right\}+\mathcal{O}(N_S^{3/2}),
\end{align}
where
\begin{equation}
    g_2(\eta,N_B)=\frac{(1+2N_B)\eta^2(2-\eta^2-2N_B(1-\eta^2))}{(1-\eta^2)(1+N_B(1-\eta^2))}-\frac{4(1+2N_B)^2\eta^2}{1+\left[1+2N_B(1-\eta^2)\right]^2}.
\end{equation}

In this regime, the optimal $\xi$ is either $\xi^{\rm opt}=1$ or $\xi^{\rm opt}=0$. The expansion holds as long as $N_S\ll N_B$. 

By setting $N_B=0$ in \eqref{exp_NSsmall} we get 
\begin{align}
    I_\eta^{\rm IF}\simeq I_{\rm shad}+4N_S\left\{1-\xi+\frac{\xi \eta^4}{1-\eta^2}\right\}+\mathcal{O}(N_S^{3/2}).
\end{align}
Here, $\xi^{\rm opt}=1$ if $\eta^4+\eta^2-1>0$, which happens for $\eta\gtrsim 0.786$. This is in contrast with what we found in Eq.~\eqref{smallN_S} in the $N_B=0$ case, {\it i.e.}, $\eta>1/\sqrt{2}\simeq 0.707$, since the latter holds in the $N_S\gg N_B$ regime.

For large $N_B$, we have that 
\begin{align}
    g_2(\eta,N_B)&= \frac{4\eta^2}{N_B(1-\eta^2)^3}\left\{\frac{3(1+\eta^2)}{N_B(1-\eta^2)}-2\right\}+\mathcal{O}(N_B^{-3}),\\
  \frac{1}{1+2N_B(1-\eta^2)}&= \frac{1}{N_B(1-\eta^2)}\left\{2-\frac{1}{N_B(1-\eta^2)}\right\}+\mathcal{O}(N_B^{-3}).
\end{align}
By setting $\eta^2=1-\varepsilon$ and solving for $\varepsilon$ small, we find that $ g_2(\eta,N_B)=\frac{1}{1+2N_B(1-\eta^2)}$ for $\eta\simeq 1-\frac{3}{2N_B}$.

\subsection*{A6: Non-commuting limits}\label{app:A6}
In the asymptotic QFI analysis, we have several situations where two limits of the QFI do not commute. Indeed, by changing their order, we get a different result. In the following, we show how to interpret this feature using an example for the squeezed-vacuum state.

Let us consider the limits $\eta\to0$ and $N_B\to0$. These two limits do not commute, as $\lim_{\eta\to0}\lim_{N_B\to0}I_{\eta}^{\rm sq}=4N_S$ and $\lim_{N_B\to0}\lim_{\eta\to0}I_{\eta}^{\rm sq}=0$. Since $I_\eta^{\rm sq}=l_2(N_S,N_B)\eta^2+\mathcal{O}(\eta^4)$ for $\eta\to0$, with $l_2(N_S,N_B)=\mathcal{O}(N_B^{-1})$, we have that the limit $\lim_{N_B\to0}\lim_{\eta\to0}I_{\eta}^{\rm sq}=0$ is valid in the $1\gg N_B\gg\eta^2$ regime. More generally, if we first set $N_B=a\eta^2$ and then we expand at $\eta=0$, we have that $I_{\eta}^{\rm sq}=\frac{4N_S^2}{a+N_S}+\mathcal{O}(\eta^2)$, and the two limit orders are retrieved by considering either $a\ll1$ or $a\gg1$.  This approach is general and can be used to solve similar scenarios. 
Qualitatively, one can say that taking one limit before the other means that the first parameter reaches the asymptotic value faster than the second one.

\section*{Appendix~B: Entanglement-assisted case}\label{app:B}
In this appendix, we discuss the entanglement-assisted case. We show the details to prove that the TMSV state is the optimal probe (Section~\hyperref[sec:IV]{IV}).

\subsection*{B1: Canonical form of the generic pure-state probe}\label{app:B1}
The following Lemma sets the canonical form of the generic mixed probe in the entanglement-assisted case.
\begin{lemma}{\bf [Canonical form of the generic probe]}
The covariance matrix and first-moment for the generic input-state of the channel \eqref{eq1}, in the case of single-mode idler, can be canonically expressed as
\begin{align}\label{canonical}
{\bf d}=\begin{bmatrix}
q \\ p \\ 0\\0
\end{bmatrix}\quad 
{\bf \Sigma}=
\begin{bmatrix}
 a{\bf S}(r) & {\bf R}(\phi){\bf C}\\
 [{\bf R}(\phi){\bf C}]^\top & b\mathbb{I}_2
\end{bmatrix},
\end{align}
where ${\bf S}(r)= {\rm diag}(r,r^{-1})$, ${\bf R}(\phi)=\begin{bmatrix}
\cos(\phi) & -\sin(\phi) \\ \sin(\phi) & \cos(\phi)
\end{bmatrix}$, ${\bf C}={\rm diag}(c_+,c_-)$. Here, all parameters are real and respect the constraints given by the Heisenberg relation ${\bf \Sigma}+\mathrm{i}{\bf \Omega}/2\succeq0$.
\end{lemma}
\begin{proof}
Let us denote by $\mathcal{E}_{x,y}$ the channel defined in  Eq.~\eqref{eq1}. We have that  $\mathcal{E}_{x,y}[(\mathcal{R}_S\otimes \mathcal{S}_I)[\rho_{SI}]]=(\mathcal{R}_S\otimes \mathcal{S}_I)[\mathcal{E}_{x,y}[\rho_{SI}]]$, where $\mathcal{R}_S$ is a generic rotation applied on the signal, $\mathcal{S}_I$ is a generic symplectic transformation applied to the idler, and $\rho_{SI}$ is a generic signal-idler state. Therefore, given a generic state $\rho_{SI}$, its covariance matrix and first-moment vector can be brought to the form in Eq.~\eqref{canonical} by applying the following operations in series. (i) Displace the idler mode in order to set  ${\bf d}_I=(0,0)^T$. (ii) Rotate the idler mode to diagonalize ${\bf \Sigma}_I$. (iii) Squeeze the idler mode to make ${\bf \Sigma}_I$ proportional to the identity. (iv) Rotate the signal to diagonalize ${\bf \Sigma}_S$. The resulting covariance matrix is ${\bf \Sigma}=\begin{bmatrix}
a{\bf S}(r) & {\bf \Sigma}_{SI} \\ {\bf \Sigma}_{SI}^\top & b\mathbb{I}_2
\end{bmatrix}$ for some ${\bf \Sigma}_{SI}$. We can decompose with the singular value decomposition, \textit{i.e.},  ${\bf \Sigma}_{SI}={\bf R}(\phi){\bf C}{\bf R}^\top(\bar\phi)$ for some $\phi$ and $\bar \phi$. Finally, (v) Apply a rotation ${\bf R}(\bar \phi)$ to the idler mode.
\end{proof}

In the following we assume $c_+\geq c_-$.
The state of Eq.~\eqref{canonical} still has too many free parameters to allow for full analytical and/or numerical treatment. We apply a similar procedure as before to constrain the parameters and use convexity of the QFI to take the optimal probe as pure. However, physicality conditions impose  constraints that we will exploit to further restrict the free parameters. To optimize the covariance matrix of the input state, we start by studying the symplectic invariant Selerian for the generic state in Eq.~\eqref{canonical}, which is
\begin{equation}
    \Delta \equiv \det{\bf \Sigma}_S+\det{\bf \Sigma}_I+2\det{\bf \Sigma}_{SI} = a^2+b^2+2c_+c_-.\label{eq:selarian}
\end{equation}
Since for a two-mode pure state $\Delta=\frac{1}{2}$~\cite{Serafini:book} and $a,b\geq\frac{1}{2}$, we have that either $c_+>0$ and $c_-<0$  or $c_+=c_-=0$. The special cases $c_+>0$ with $c_-=0$ and $c_-<0$ with $c_+=0$, do not allow for a positive definite covariance matrix. 

\begin{lemma}\label{lem:aeqb}
A pure-state on the form of Eq.~\eqref{canonical} has $a=b$.
\end{lemma}

\begin{proof} If $c_+=c_-=0$, then by Eq.~\eqref{eq:selarian} and $\Delta = \frac{1}{2}$ we have that $a=b=\frac{1}{2}$. Assume, instead,  $c_+>0$. Further, assume for now that $\phi =0$ and $r=1$.  Using $\Delta=\frac{1}{2}$ with Eq.~\eqref{eq:selarian}, we solve for $c_-$ as
\begin{equation}
    c_- = \frac{1}{2c_+}\left(\frac{1}{2}-a^2-b^2\right).\label{eq:cminus}
\end{equation}
Fixing the Selarian is  not sufficient for purity. In fact, the determinant of the covariance matrix is
\begin{equation}
\det{\bf \Sigma} = \left(ab-c_-^2\right)\left(ab-c_+^2\right).\label{eq:determinant}
\end{equation}
We use Eq.~\eqref{eq:cminus} and purity with Eq.~\eqref{eq:determinant} to solve for $c_+^2$
such that
\begin{equation}
    c_+^2 = \frac{1}{8ab}\left[a^2\left(a^2+3b^2-1\right)+b^2\left(b^2+3a^2-1\right)\pm(a^2-b^2)\sqrt{\left(\left(a+b\right)^2-1\right)\left(\left(a-b\right)^2-1\right)}\right],\label{eq:cplus}
\end{equation}
where there is an apparent choice of sign depending on the relation between $a$ and $b$. However, the ambiguity is resolved by recognizing that $c_+\in \mathbb{R}$. Since $a,b\geq\frac{1}{2}$ implies $a^2+3b^2\geq 1$ and $b^2+3a^2\geq 1$, reality of $c_+$ depends only on the  square root of Eq.~\eqref{eq:cplus}. This requires that \textit{either}  $a=b$, or
\begin{equation}
    \left[\left(a+b\right)^2-1\right]\cdot\left[\left(a-b\right)^2-1\right] \geq 0.
\end{equation}
which reduces to $a+1\leq b$ or $b\leq a-1$ to ensure $c_+$ is real. However, pure states with $a+1\leq b$ or $b\leq a-1$ are non-physical, with $ab<c_+^2$, since the covariance matrix would not be positive definite. Thus, the  only valid choice is $a=b$. This results holds  for arbitrary $\phi$ and $r$. Application of a rotation $\phi$ followed by squeezing $r$, \textit{i.e.},  ${\bf S}(r){\bf R}(\phi)$, to the signal is an invertible purity-preserving transformation that does not affect ${\bf \Sigma}_I$.
\end{proof} 

We are now entitled to prove Lemma~\ref{prop:covariance}.
\begin{proof}[Proof of Lemma~\ref{prop:covariance}]
Assume  $\phi=0$ and $r=1$.  By Lemma~\ref{lem:aeqb} we substitute $a=b$ in  Eq.~\eqref{eq:cminus} and Eq.~\eqref{eq:cplus} to find  $c_+ =-c_- = \sqrt{a^2-\frac{1}{4}}$. Now the covariance matrix of the probe state is parametrised by $a$ alone, as
\begin{equation}
    {\bf \Sigma}_S = {\bf \Sigma}_I = {\rm diag}\left(a,a\right), \quad {\bf \Sigma}_{SI} = \sqrt{a^2-\frac{1}{4}} \, {\rm diag}\left(1,-1\right).
\end{equation}
Application of ${\bf S}(r){\bf R}(\phi)$ to the signal gives the stated covariance matrix.
 \end{proof}  

\subsection*{B3: Two-mode QFI}\label{app:B2}
The QFI with the support of an entangled ancilla mode is computed from Eq.~\eqref{eq:QFIgeneral} for the canonical two-mode pure-state probe according to Eq.~\eqref{eq:purecovariance} with displacement $\mathbf{d} = \sqrt{2N_{\rm coh}}(\cos\theta,\sin\theta,0,0)^\top$, transformed as  Eqs.~\eqref{eq:displdyn}--\eqref{eq1}. Explicitly, the expression rather lengthy, but we include it for completeness as
\begin{equation}
    I_\eta^{\rm EA} = I_{\eta,1}^{\rm EA} + I_{\eta,2}^{\rm EA} ,
\end{equation}
where
\begin{align}
I_{\eta,1}^{\rm EA} =\quad &    \mathrm{Tr}\left[ {\bf L}_2
	\left(\partial_\eta 
	{{\bf \tilde \Sigma}}\right) \right], \\
	=\quad & \frac{4a^2 + 1}{\eta^2}+\frac{2\eta^2}{\left(1-\eta^2\right)^2}+\frac{r}{\eta^2\left[1- N_B\left(N_B+1\right)\left(4a^2-1\right)\left(1-\eta^2\right)^2 \right]} \times\nonumber  \\
    & \times\left\{ \frac{ \left(1-\eta^2\right)N_B\left(N_B + 1\right)\left[4a^2\left(1-\eta^2\right)  + \eta^2 + 1\right]^2\left[4a^2\left(1+2N_B\right)^2 - 1\right]}{2a\eta^2\left(1 + r^2\right)\left(1 + 2N_B\right)+r\left[ \left(1- \eta^2\right)\left(4a^2\left(1+2N_B\right)^2 -1\right)-2\eta^2 \right]} \right.\nonumber \\
    &\left.- \frac{\left[\left(1-\eta^2\right)\left(4a^2 + 1\right)\left(2N_B^2 + 2N_B + 1\right)+2\eta^2\left(1-\eta^2\right)^{-1} \right]^2 - 16a^2\eta^4\left(2N_B + 1\right)^2}{2a\eta^2\left(1+r^2\right)\left(1+2N_B\right)\left(1-\eta^2\right) + r\left[  \left(1-\eta^2\right)^2\left(4a^2 + 1\right)\left(2N_B^2 + 2N_B + 1\right)+2\eta^2  \right]} \right\},
\end{align}
and
\begin{align}    
I_{\eta,2}^{\rm EA} =\quad &  \left(\partial_\eta \tilde {\mathbf{d}}\right)^\top {\bf \tilde \Sigma}^{-1}  \left(\partial_\eta \tilde {\mathbf{d}}\right) \\
=\quad & 8N_{\rm coh}a\left\{\frac{\cos^2\theta}{2a\left(1-\eta^2\right)\left(1+2N_B\right)+r\eta^2}+\frac{\sin^2\theta}{2a\left(1-\eta^2\right)\left(1+2N_B\right)+\eta^2/r}\right\}.
\end{align}
This expression, indeed, does not depend on $\phi$ (introduced in Lemma~\ref{prop:covariance}).

\section*{Appendix~C: Normalized background, $N_B\rightarrow N_B/(1-\eta^2)$}\label{app:C}

With the normalization  $N_B\rightarrow N_B/(1-\eta^2)$, the channel has a constant background noise for all transmissions.  Let us denote the QFI under this change of variables by $I^{\rm norm}_\eta$. In this case, there is no shadow-effect, as $\mathcal{I}^{\rm norm}_\eta(N_S=0)=0$. 
The idler-free QFI is
\begin{equation}
   {I}_\eta^{\rm norm, IF}\left(N_{\rm coh},N_{\rm sq}\right) = {I}_{\eta,1}^{\rm norm, IF}\left(N_{\rm coh},N_{\rm sq}\right) + {I}_{\eta,2}^{\rm norm, IF}\left(N_{\rm sq}\right) ,
\end{equation}
where
\begin{align}
    {I}_{\eta,1}^{\rm norm, IF}\left(N_{\rm coh},N_{\rm sq}\right) &= \frac{4N_{\rm coh}}{r\eta^2+ 2N_B+1-\eta^2}, \\
  {I}_{\eta,2}^{\rm norm, IF}\left(N_{\rm sq}\right) &= \frac{4N_{\rm sq}\eta^2}{B}\left\{ \frac{\left(N_{\rm sq}+1\right)\left(2N_B+1\right)^2}{2B+1}-1 \right\},
\end{align}
with $B = N_B\left(N_B+1\right)+N_{\rm sq}\eta^2\left(2N_B+1\right)-N_{\rm sq}\eta^4$. 

Similarly, the ancilla-assisted QFI using the TMSV as a probe is 
\begin{equation}
    {I}_\eta^{\rm norm, TMSV} = \frac{4N_S\left( N_B+1+ N_S\left(N_B+1-\eta^2\right)\right)}{\left[N_B+1-\eta^2\right]\left[ N_B+1+ N_S\left(2N_B+1-\eta^2\right)\right]}.\label{altchanneltmsv}
\end{equation}

Notice that both QFIs are the same as the unnormalized case for $N_B=0$. The total QFI $\mathcal{I}_\eta^{\rm norm}$ can be computed by just using the relation $N_S=\mathcal{N}_S/M$.


\begin{thebibliography}{99}
\bibitem{Footnote2}{Notice that to the first order in $\eta^2$ the TMSV state performs better than an optimized displaced squeezed state, as $g_1<1$ in Eq.~\eqref{eq:g1}.}
\bibitem{Footnote3}{If $N_S^{\rm opt}(x)=0$, then one can first set $N_S^{\rm opt}(x)=\varepsilon$, and then optimize for $x\in[0,\varepsilon]$. Finally, $\varepsilon$ is sent to zero.}
\bibitem{Footnote4}{In Ref.~\cite{Nair2020} the authors claim that the bound in Eq.~(19) of their paper needs not to be saturable. Here, we show that it is actually saturated by a TMSV probe.}
\end{thebibliography}
\end{document}